\pdfoutput=1
\RequirePackage{ifpdf}
\ifpdf 
\documentclass[pdftex]{sigma}
\else
\documentclass{sigma}
\fi

\usepackage{pgf,tikz,pgfplots,tikz-3dplot}
\usetikzlibrary{calendar,folding}
\usetikzlibrary{arrows,patterns,decorations.pathmorphing,backgrounds,positioning,fit,petri}
\usetikzlibrary{calc,3d,intersections,shapes}
\usepgfplotslibrary{polar}
\usetikzlibrary {arrows.meta}
\usetikzlibrary{shapes.misc}

\numberwithin{equation}{section}

\newtheorem{Theorem}{Theorem}[section]
\newtheorem*{Theorem*}{Theorem}

\newtheorem{Lemma}[Theorem]{Lemma}
\newtheorem{Proposition}[Theorem]{Proposition}
 { \theoremstyle{definition}
\newtheorem{Definition}[Theorem]{Definition}

\newtheorem{Remark}[Theorem]{Remark} }

\begin{document}

\allowdisplaybreaks

\newcommand{\arXivNumber}{2404.13688}

\renewcommand{\PaperNumber}{087}

\FirstPageHeading

\ShortArticleName{Algebraic Complete Integrability of the $a_4^{(2)}$ Toda Lattice}

\ArticleName{Algebraic Complete Integrability\\ of the $\boldsymbol{a_4^{(2)}}$ Toda Lattice}

\Author{Bruce Lionnel LIETAP NDI~$^{\rm a}$, Djagwa DEHAINSALA~$^{\rm b}$ and Joseph DONGHO~$^{\rm a}$}

\AuthorNameForHeading{B.L.~Lietap Ndi, D.~Dehainsala and J.~Dongho}

\Address{$^{\rm a)}$~University of Maroua, Faculty of Sciences, Department of Mathematics Computer Sciences,\\
\hphantom{$^{\rm a)}$}~P.O.~Box 814, Maroua, Cameroon}
\EmailD{\href{mailto:nbruce.lionnel@gmail.com}{nbruce.lionnel@gmail.com}, \href{mailto:josephdongho@yahoo.fr}{josephdongho@yahoo.fr}}

\Address{$^{\rm b)}$~Department of Mathematics, Faculty of Exact and Applied Sciences, University of NDjamena,\\
\hphantom{$^{\rm b)}$}~1 route de Farcha, P.O.~Box 1027, NDjamena, Chad}
\EmailD{\href{mailto:djagwa73@gmail.com}{djagwa73@gmail.com}}

\ArticleDates{Received April 25, 2024, in final form September 25, 2024; Published online October 05, 2024}

\Abstract{The aim of this work is focused on the investigation of the algebraic complete integrability of the Toda lattice associated with the twisted affine Lie algebra \smash{$a_4^{(2)}$}. First, we prove that the generic fiber of the momentum map for this system is an affine part of an abelian surface. Second, we show that the flows of integrable vector fields on this surface are linear. Finally, using the formal Laurent solutions of the system, we provide a~detailed geometric description of these abelian surfaces and the divisor at infinity.}


\Classification{34G20; 34M55; 37J35}

\section{Introduction}
The study of integrable Hamiltonian systems has been motivated by
several factors, including the development of powerful and beautiful
mathematical theories and the application of integration concepts to
various physical, biological and chemical systems. However, it
remains challenging to describe or recognize integrable Hamiltonian
systems with ease, as they are exceptional cases. The Korteweg--de Vries
equation has generated numerous new ideas in the field of
completely integrable Hamiltonian systems, leading to unexpected
connections between mechanics, spectral theory, Lie algebra theory,
algebraic geometry, and even differential geometry. Some interesting
integrable systems also appear as coverings of algebraic complete
integrable systems. These systems are sometimes also called
algebraic complete integrable.

 An algebraic complete integrable system can be linearized on a complex torus, and its
invariant functions (often called first integrals or constants) are
polynomial maps and their restrictions to an invariant complex
variety are meromorphic functions on a complex abelian variety. The
fluxes generated by the constants of motion are straight lines in
this complex abelian variety.

Several nonlinear completely integrable systems were known in the
19th century, among them the geodesic flow on the ellipsoid,
Neumann's system or the Kowalevski top. The Toda lattice, introduced
by Morikazu Toda in 1967, is a simple model for a one-dimensional
crystal in solid-state physics. It is famous because it is one of
the first integrable systems for which a Lax pair was discovered by
Flaschka. The classical Toda lattice is a system of particles with
unit mass, connected by exponential springs. Its equations of motion derived from the Hamiltonian
\begin{equation}
 H=\frac{1}{2}\sum_{j=1}^n p_j^2+ \sum_{j=1}^{n-1}
 {\rm e}^{q_j-q_{j+1}}, \label{hamilttodaclasnperio}
\end{equation}
where $q_j$ is the position of the $j$-th particle and
$p_j$ is its amount of movement. This type of Hamiltonian was
considered first by Morikazu Toda \cite{toda,toda1}. The equation
\eqref{hamilttodaclasnperio} is known as the finite classic non-periodic Toda lattice to distinguish other versions of various forms
of the system. The periodic version of \eqref{hamilttodaclasnperio}
is given by
\[H=\frac{1}{2}\sum_{j=1}^n p_j^2+ \sum_{j=1}^n {\rm e}^{q_j-q_{j+1}}, \qquad q_{n+1} =q_1,\]
 where the equations of motion are given by
\[
\dot{p}_j=-\frac{\partial H}{\partial q_j}=
{\rm e}^{(q_{j-1}-q_j)}-{\rm e}^{(q_j-q_{j+1})},\qquad \dot{q}_j=\frac{\partial H }{\partial p_j}=p_j, \qquad 1\leq j\leq n.
\]
The integrability of the periodic Toda lattice was established by
Henon \cite{henon} and Flaschka \cite{flaschka} using the Lax pairs
method. In 1976, Bogoyavlensky \cite{bogoya} introduced a
generalization of the classical Toda periodic lattice to arbitrary
Lie algebras.

 Adler, van~Moerbeke and Vanhaecke in \cite{al_Moer3} give the explicit
Hamiltonians for the periodic Toda lattices that involve precisely
three (connected) particles. In two-dimensional, there are precisely six cases of them, going with the extended root systems
\smash{$a_2^{ (1 )}$}, \smash{$a_4^{(2)}$},
\smash{$c_2^{ (1 )}$}, \smash{$d_3^{(2)}$},
\smash{$g_2^{ (1 )}$} and~\smash{$d_4^{ (3 )}$}. They prove in \cite{al_Moer5,al_Moer3}
that the case \smash{$a_2^{ (1 )}$} is algebraic completely integrable. In his case, Dehainsala prove in~\cite{Dehainsala} that the two cases \smash{$c_2^{ (1 )}$} and \smash{$d_3^{(2)}$} are algebraic complete integrable.

In this work, we consider that \smash{$a_4^{(2)}$} is a
two-dimensional integrable system. This system satisfies the
linearization criterion \cite[Theorem 6.41]{al_Moer3}. We prove
that this system is an algebraic completely integrable in the
Adler--van Moerbeke sense.

To prove this, with respect to the complex Liouville theorem,
firstly, by the indicial locus and the Kowalevski matrix, we have
shown that our system has three distinct families of homogeneous
Laurent solutions with weights depending on four free parameters and
find the Zariski open set~$\Omega$ and the fiber $\mathbb{F}_c$,
$c \in \Omega$. Secondly, with these Laurent solutions, we determined
the Painlev\'{e}~divisor and their arithmetic genus. We have obtain
three divisors, for $c \in \Omega$, the Painlev\'e divisor~\smash{$\Gamma_c^{(0)}$} is a~smooth genus three hyperelliptic
curve, the Painlev\'e divisor~\smash{$\Gamma_c^{(1)}$} is a~smooth genus four curve and the Painlev\'e divisor
\smash{$\Gamma_c^{(2)}$} is a~smooth genus two hyperelliptic
curve. Thirdly, to compact the fiber, we determined the projective
space to embedding the divisor at infinity to find the singularities
and the intersection between the different curves. We have obtain
twenty-five $(25)$ functions which forms the basis of the
projective space and determined the intersections points. To end our
prove, we show that the vector field
$(\varphi_c)_{\ast}\mathcal{V}_1$ extends to an
holomorphic vector fields on $\mathbb{P}^{24}$. To show that the
vector field $(\varphi_c)_{\ast}\mathcal{V}_1$ is
holomorphic on two chart of $\mathbb{P}^{24}$, we have established
that this vector field can be written as a quadratic vector field in
two appropriate chart.\looseness=-1

This paper is organized as follows. In Section~\ref{sec2}, we review the
basic notions of algebraic integrability in sense of Adler--van Moerbeke. Section~\ref{sec3} contains the main part of the paper, we verify
that the \smash{$a_4^{(2)}$} Toda lattice is Liouville
integrable, we do the Painlev\'{e} analysis of the system. This
analysis shows that our integrable system admits three principal
balances, i.e., three families of Laurent solutions depending on the
maximal number of free parameters,four in our case. Thus, by
confining each family of Laurent solutions to the invariant
manifolds we calculate the Painlev\'{e} divisors associated to these
principal balances and, for $c\in \Omega$, we give an explicit
embedding of the invariant manifold $\mathbb{F}_c$ in the
projective space $\mathbb{P}^{24}$. Finally, in~Section~\ref{sec4}, we
determine the holomorphic differentials forms on the abelian
surface.

\section{Preliminaries}\label{sec2}
In this section, we also recall some basics notions. For more
comprehension just read the book~\cite{al_Moer3}.

 Let $v =
(v_1, \dots, v_n)$ be a collection of positive integers
without a common divisor. Such a $v$ is called a weight vector.
\begin{Definition}[{\cite{al_Moer3}}]
 A polynomial $f\in \mathcal{F}(\mathbb{C}^n):=\mathbb{C}[x_1,\dots,x_n]$ is a weight homogeneous polynomial of weight
$k$ (with respect to $v$) if
$f(t^{v_1}x_1,\dots,t^{v_n}x_n)=t^kf(x_1,\dots,x_n)$, $\forall \, (x_1,\dots,x_n)\in
\mathbb{C}^n$ and $t\in \mathbb{C}$.
\end{Definition}
When the weight of $f$ is $k\in \mathbb{N}$, we denote
\[\mathcal{F}^{(k)}:= \{F\in\mathcal{F}(\mathbb{C}^n)\mid \varpi(F) = k
\}.\]

\begin{Definition}[{\cite{al_Moer3}}]
A polynomial vector field on $\mathbb{C}^n$
\begin{equation}\label{polyvect}
\begin{cases}
\dot{x}_1 = f_1(x_1,\dots, x_n),\\
 \cdots\cdots\cdots\cdots\cdots\cdots\cdot \\
\dot{x}_n = f_n(x_1,\dots, x_n)
 \end{cases}
\end{equation}
is called a weight homogeneous vector field of weight $k$ (with
respect to $v$) if each of the polynomials $f_1,\dots,f_n$ is
weight homogeneous (with respect to $v$) and if
$\varpi(f_i )=v_i + k =\varpi(x_i)+ k$ for
$i=1,\dots,n$.
\end{Definition}

According to \cite[Proposition 7.6]{al_Moer3}, if \eqref{polyvect}
is a weight homogeneous vector field, then Laurent solutions have the
form
\[x_i(t)=\frac{1}{t^{v_i}}\sum_{k=0}^\infty x_i^{(k)}t^k,
\qquad i=1,\dots,n,
\quad \text{with} \quad x^{(0)}=\bigl(x_1^{(0)},\dots,x_n^{(0)}\bigr)\neq0,\]
 are called weight homogeneous Laurent solution. We will say that a formal Laurent solution is a principal
balance if it depends on $n-1$ free parameters; otherwise, it will
be called the lower balance.

The positive integers $v_i$ being the weights of the phase
variables, then the leading coef\-ficients~\smash{$x_i^{(0)}$}
satisfy the nonlinear algebraic equations
\begin{equation}\label{equaindicial}
 v_ix_i^{(0)}+f_i\bigl(x_1^{(0)},x_2^{(0)},\dots,x_n^{( 0)}\bigr)=0, \qquad i=1,\dots,n,
\end{equation}
while the other coefficients \smash{$x_i^{(k)}$},
$k\geq1$, satisfy
\[\bigl(k\operatorname{Id}_n-\mathcal{K}\bigl(x^{(
0)}\bigr)\bigr)x^{(k)}=P^{(k)},\]
 where
 \[x^{(k)}=\bigl(x_1^{(k)},x_2^{(k)},\dots,x_n^{(
k)}\bigr)^{\mathsf{T}},
\qquad P^{(k)}=\bigl(P_1^{(k)},P_2^{(k)},\dots,P_n^{(
k)}\bigr)^{\mathsf{T}}
\]
and each \smash{$P_i^{(k)}$} is a polynomial which depends only on the variables \smash{$x_1^{(l)},x_2^{(l)},\dots,
x_n^{(l)}$} with ${0\leq l\leq k}$. The matrix
$\mathcal{K}$ of order $n$ is defined by
\smash{$\mathcal{K}_{ij}=\frac{\partial f_i}{\partial x_j}+v_i
\delta_{ij}$} is called the Kowalevski matrix.

 Let $\Lambda$ be a discrete subgroup of rank $k$ of
$\mathbb{C}^n$. The quotient group $\mathbb{C}^n/ \Lambda$ has a
smooth complex manifold structure induced by the projection
$\pi\colon \mathbb{C}^n\longrightarrow \mathbb{C}^n/ \Lambda$. This
complex manifold is compact if and only if $k=2n$. In this case,
$\mathbb{T}^n=\mathbb{C}^n/ \Lambda$ is a complex torus. If such a
complex torus is also an algebraic variety, i.e., admits a projective
embedding, we then say that it is an abelian variety. If $k=2n=4$,
we say that it is an abelian surface.

To define an algebraic complete integrable system, we need the
following definitions.

\begin{Definition}[\cite{al_Moer3}] A fiber of $\mathbf{F}$ is a common level set of the functions
$\mathbf{F_i}$. The fiber of $\mathbf{F}$ that passes through
$m\in M$ will be denoted by \[\mathbf{F_m}=\{p\in
M\mid \mathbf{F_i}(p)=\mathbf{F_i}(m),\, \forall i
\in \{1,\dots,s\}\}.\]
\end{Definition}

For $c\in \mathbb{R}^s$ (resp.\ $\mathbb{C}^s$), we will note the
fiber $\mathbf{F}^{-1}(c)$ above $c$ by
$\mathbf{F}_c$. So, we have
\[\mathbf{F}_m=\mathbf{F}^{-1}(\mathbf{F}(m))=\mathbf{F}_{\mathbf{F} (m)}\qquad \forall m \in M.\]
The set of regular values of $\mathbf{F}$ is a residual
subset (hence a dense subset) of $\mathbb{R}^s$ (resp.\
$\mathbb{C}^s$). By~the inverse function theorem, the fiber
$\mathbf{F}_c$ over each regular value $c$ that lies in the
image of~$\mathbf{F}$ is non-singular. Hence, when
$\mathbf{F}=(F_1,\dots,F_s)$ is involutive, the
Hamiltonian vector fields~$\mathcal{X}_{\mathbf{F}_i }$,
$1\leq i \leq s$, commute and for any point $m$ they
are tangent to the non-singular affine part of~$\mathbf{F}(m)$.

\begin{Definition}
An abelian variety is a complex torus $\mathbb{C}^r/\Lambda$
$(\Lambda \mbox{ a lattice in } \mathbb{C}^r)$ which is
projective, which means that it admits an embedding in a project
space $\mathbb{P}^N$.
\end{Definition}

An abelian variety $\mathbb{T}^r$ will be called an
\textit{irreducible abelian variety}, when $\mathbb{T}^r$ does not
contain any abelian subvariety, otherwise it will be called a
reducible abelian variety.

\begin{Definition}
Let $(M,\{\cdot,\cdot\}, \mathbf{F})$ be a complex
integrable system, where $M$ is a non-singular affine variety and
where $\mathbf{F}=(F_1,\dots,F_s)$. We say that
$(M,\{\cdot,\cdot\}, \mathbf{F})$ is an algebraic completely integrable system if for generic $c
\in \mathbb{C}^s$ the fiber~$\mathbf{F}_c$ is an affine part of
an abelian variety and if the Hamiltonian vector fields
$\mathcal{X}_{\mathbf{F}_i } $ are translation invariant, when
restricted to these fibers. In the particular case in which $M$ is
an affine space~$\mathbb{C}^n$, we will call
$(\mathbb{C}^n,\{\cdot,\cdot\}, \mathbf{F})$ a~polynomial algebraic complete integrable system. When the generic abelian variety of the
algebraic complete integrable system is irreducible, we speak of an irreducible algebraic complete integrable
system.\looseness=1
\end{Definition}

 An integrable system is said to be algebraic completely integrable
if the fibers of the momentum map are affine parts of abelian varieties and the integrable fields are linear.

The following theorem gives a necessary condition for the algebraic
integrability of an integrable system. It is inspired by the
Kowalevski work \cite{Kowaleski}, and is based on the fact that the
phase space of an algebraic complete integrable system admits a partial compactification on
which the integrable vector fields extend into complete vector
fields. This means that each of the integrable vector fields of an
irreducible algebraic complete integrable system on $\mathbb{C}^n$ admits one or several
families of Laurent solutions (called balances), which will lead to
a necessary condition for algebraic complete integrability, which we
call the Kowalevski--Painlev\'e criterion.\looseness=1

\begin{Theorem}[Kowalevski--Painlev\'e criterion {\cite{al_Moer3}}]
Let $ (\mathbb{C}^n,\{\cdot,\cdot\}, \mathbf{F} )$ be an
irreducible polynomial algebraic complete integrable system, where
$\mathbf{F}= (F_1,\dots,F_s )$ is a family of polynomials
and $ (x_1,\dots,x_n )$ is a system of linear
coordinates on $\mathbb{C}^n$. Let $\mathcal{V}$ be any one of
the integrable vector fields
$\mathcal{X}_{F_1},\dots,\mathcal{X}_{F_s}$. For each $1\leq i
\leq n$ such that $x_i$ is not constant along the integral curve
of~$\mathcal{V}$, i.e., $\dot{x}_i:=\mathcal{V}[x_i]\neq 0$,
there exists a principal balance
$x (t )= (x_1 (t ), \dots,
x_n (t ) )$, depending on $n-1$ free parameters
for which $x_i (t )$ has a pole.
\end{Theorem}
Let $\mathcal{A}$ an affine variety,
$\varphi\colon \mathcal{A}\longrightarrow \mathbb{P}^N$ a regular map,
 let \smash{$\overline{\varphi(\mathcal{A} )}$} be the closure of the
image of~$\mathcal{A}$ in~$\mathbb{P}^N$.

The following theorem gives the sufficient conditions to be
satisfied by the fibers of the momentum map of an integrable system
to be algebraic complete integrable.
\begin{Theorem}[complex Liouville theorem {\cite{al_Moer3}}]\label{liouville}
Let $\mathcal{A} \in \mathbb{C}^s$ be a non-singular affine
variety of dimension $r$ which supports $r$ holomorphic vector
fields $\mathcal{V}_1,\dots,\mathcal{V}_r$ and let
$\varphi\colon \mathcal{A}\longrightarrow \mathbb{C}^N\subset
\mathbb{P}^N$ be a regular map; here $\mathbb{C}^N\subset
\mathbb{P}^N$ is the usual inclusion of $\mathbb{C}^N$ as the
complement of a hyperplane~$\mathbf{H}$ in~$\mathbb{P}^N$. We
define \smash{$\Delta:=\overline{\varphi (\mathcal{A} )}\setminus
\varphi (\mathcal{A} )$} and we decompose the analytic
subset $\Delta$ as $\Delta=\Delta'\cup \Delta''$, where
$\Delta'$ is the union of the irreducible components of $\Delta$
of dimension $r-1$ and $\Delta''$ is the union of the other
irreducible components of $\Delta$. The following conditions are
assumed to be verified:
\begin{enumerate}\itemsep=0pt
\item[$1.$] $\varphi\colon \mathcal{A}\longrightarrow \mathbb{C}^N$ is an isomorphic embedding.
 \item[$2.$] The vector fields commute pairwise,
$[\mathcal{V}_i,\mathcal{V}_j]=0$ for $1\leq i,j \leq r$.
 \item[$3.$] At every point $m\in \mathcal{A}$, the vector fields
$\mathcal{V}_1,\dots,\mathcal{V}_r$ are independent.
 \item[$4.$] The vector field $\varphi_{\ast}\mathcal{V}_1$
extends to a vector field $\overline{\mathcal{V}}_1$ which is
holomorphic on a neighborhood of $\Delta'$ in $\mathbb{P}^N$.
 \item[$5.$] The integral curves of $\overline{\mathcal{V}}_1$ that start at points $m\in \Delta'$ go immediately
into $\varphi (\mathcal{A} )$.
\end{enumerate}
Then \smash{$\overline{\varphi (\mathcal{A} )}$} is an abelian
variety of dimension $r$ and $\Delta''=\varnothing$, so that
\smash{$\overline{\varphi (\mathcal{A} )}=\varphi (\mathcal{A} )
\cup \Delta'$}. Moreover, the vector fields $\varphi_{
\ast}\mathcal{V}_1,\dots,\varphi_{\ast}\mathcal{V}_r$ extends to
holomorphic vector fields on
\smash{$\overline{\varphi (\mathcal{A} )}$}.
\end{Theorem}

\section[Algebraic integrability of the a\_4\^{}(2) Toda lattice]{Algebraic integrability of the $\boldsymbol{a_4^{(2)}}$ Toda lattice}\label{sec3}

The aims of this section is to prove, by following the strategy
developed by Adler, van~Moerbeke and Vanhaecke in \cite[Section~9.4]{al_Moer3} to establish the algebraic integrability of
the \smash{$a_2^{ (1 )}$} Toda lattice, the algebraic complete
integrability of the~\smash{$a_4^{(2)}$} Toda lattice.

\subsection[Liouville integrability of the a\_4\^{}(2) Toda lattice]{Liouville integrability of the $\boldsymbol{a_4^{(2)}}$ Toda lattice}

The differential equations of the periodic \smash{$a_4^{(2)}$}
Toda lattice
 are given on the five dimensions hyperplane
 $\mathcal{H} =\big\{( x_0, x_1, x_2, y_0, y_1, y_2 )\in
\mathbb{C}^6 \mid y_0 + 2 y_1 + 2y_2 = 0 \big\}$ of $\mathbb{C}^6$ by
\[
 \begin{cases}
 \dot{x}=x.y,\\
 \dot{y}=Ax,
 \end{cases}\]
where $x=(x_0,x_1,x_2)^{\mathsf{T}}$,
$y=(y_0,y_1,y_2)^{\mathsf{T}}$ and $A$ is the Cartan matrix
of the twisted affine Lie algebra \smash{$a_4^{(2)}$} given in
\cite{al_Moer3} by
\[
 \begin{pmatrix}
 \hphantom{-} 2 & -2 & \hphantom{-}0 \\
 -1 & \hphantom{-}2 & -2 \\
 \hphantom{-}0 & -1 & \hphantom{-}2 \\
 \end{pmatrix}
\]
and $\varepsilon=(1,2,2)^{\mathsf{T}}$ is the
normalized null vector of $A^{\mathsf{T}}$.
 The equations of motion of the \smash{$a_4^{(2)}$} Toda lattice are given
 in \cite{al_Moer3} by
\begin{alignat}{3}
& \dot{x}_0=x_0y_0, \qquad && \dot{y}_0=2x_0-2x_1, & \nonumber\\
& \dot{x}_1=x_1y_1, \qquad && \dot{y}_1=-x_0+2x_1-2x_2, &\nonumber\\
& \dot{x}_2=x_2y_2, \qquad && \dot{y}_2=-x_1 + 2x_2. & \label{systa4}
\end{alignat}
We denote by $\mathcal{V}_1$ the vector field defined by the above
differential equations \eqref{systa4}. Then $\mathcal{V}_1$ is the
Hamiltonian vector field, with Hamiltonian function $ F_2=y_0^2
+ 4y_2^2-4x_0-8x_1-16x_2 $ with respect to the Poisson
structure $\{\cdot,\cdot\}$ defined by the following
skew-symmetric matrix:
\begin{equation}\label{matrij2}
 J=\frac{1}{8}\begin{pmatrix}
 0 & 0 & 0 & 4x_0 & -2x_0 & 0 \\
 0 & 0 & 0 & -2x_1 & 2x_1 & -x_1 \\
 0 & 0 & 0 & 0 & -x_2 & x_2 \\
 -4x_0 & 2x_1 & 0 & 0 & 0 & 0 \\
 2x_0 & -2x_1 & x_2 & 0 & 0 & 0 \\
 0 & x_1 & -x_2 & 0 & 0 & 0
 \end{pmatrix}.
\end{equation}
This Poisson structure is given on $\mathbb{C}^6$; the function
$F_0=y_0 + 2y_1 + 2y_2$ is a Casimir, so that the hyperplane
$\mathcal{H}$ is a Poisson subvariety. The rank of this Poisson
structure $\{\cdot,\cdot\}$ is $0$ on the three-dimensional
subspace $\{x_0 = x_1 = x_2 = 0\}$; the rank is $2$ on the three
four-dimensional subspaces: $\{x_0 = x_1 = 0\}$, $\{x_0 = x_2 =
0\}$ and $\{x_1 = x_2 = 0\}$. Thus, for all points of
$\mathcal{H}$ except the four subspaces above the rank is $4$.
The vector field $\mathcal{V}_1$ admits also the following three
constants of motion:
\begin{gather}
 F_1=x_0x_1^2x_2^2, \nonumber\\
 F_2=y_0^2 + 4y_2^2-4x_0-8x_1-16x_2, \nonumber\\
 F_3= \bigl(y_0^2-4x_0\bigr)\bigl(y_2^2-4x_2\bigr)-4x_1(y_0y_2-4x_2-x_1).\label{invar1a4}
\end{gather}
$F_1$ is a Casimir for $\{\cdot,\cdot\}$, and the function
$F_3$ generates a second Hamiltonian vector field
$\mathcal{V}_2$, which commutes with $\mathcal{V}_1$, given by
the differential equations
\begin{gather}
x_0^{\prime}=x_0y_2(y_0y_2-2x_1)-4x_0x_2y_0, \nonumber\\
x_1^{\prime}=-x_1y_1y_2(y_1+y_2) - x_1^2y_1+x_1(x_0y_2+2x_2y_0), \nonumber\\
x_2^{\prime}=x_2(y_1+y_2)((y_1+y_2)y_2+x_1)+x_0x_2y_0, \nonumber\\
y_0^{\prime}=2\bigl(2x_1x_2+x_0y_2^2\bigr)+ x_1(2x_1-y_0y_2)-8x_0x_2, \nonumber\\
y_1^{\prime}=-x_0y_2^2+2x_2 ( 3x_0-x_1 ) +y_0y_2 (
x_1+x_2 )-2x_1^2+x_2y_0y_1, \nonumber\\
y_2^{\prime}=x_1y_2( y_1+y_2)+x_1^2-x_2( y_1+y_2)
-2x_2x_0. \label{syst3}
\end{gather}
Hence, the system \eqref{systa4} is completely integrable in the
Liouville sense. It can be written as a~Hamiltonian vector fields
\[\dot{z}=J\frac{\partial H}{\partial z},\qquad z=(z_1,\dots,z_6)^{\mathsf{T}}=( x_0
, x_1, x_2, y_0, y_1, y_2 )^{\mathsf{T}},\]
where $H=F_2$. The
Hamiltonian structure is defined by the following Poisson bracket:
\[ \{F,H\}=\biggl\langle\frac{\partial F}{\partial z}, J\frac{\partial H}
{\partial z}\biggr\rangle=\displaystyle\sum_{i,k=1}^6
J_{ik}\frac{\partial F}{\partial z_i}\frac{\partial H}{\partial
z_k},
\] where $\frac{\partial H}{\partial z}=\bigl(\frac{\partial
H}{\partial x_0},\frac{\partial H}{\partial x_1},\frac{\partial
H}{\partial x_2},\frac{\partial H}{\partial y_0},\frac{\partial
H}{\partial y_1},\frac{\partial H}{\partial y_2}\bigr)^{\mathsf{T}}$ and
$J$ is an
antisymmetric matrix.

The vector field $\mathcal{V}_2$ admits the same constants of
 motion \eqref{invar1a4} and is in involution with $\mathcal{V}_1$ therefore $\{F_2,F_3\}=0$. The involution $\sigma$ defined on $\mathbb{C}^6$ by
 \[\sigma(x_0,x_1,x_2,y_0,y_1,y_2)=(x_0,x_1,x_2, -y_0,-y_1,-y_2)
 \]
 preserves the constants of motion $F_1$, $F_2$ and $F_3$, hence leave the fibers of the momentum map~$F$ invariant. This involution can be restricts to the hyperplane $\mathcal{H}$.

Let $\mathbf{F}=(F_1,F_2,F_3)\colon \mathcal{H} \rightarrow
\mathbb{C}^3$ be the momentum map; functions $\mathbf{F}_i$ being
two by two in involution, $\mathbf{F}$ is involutive. The Jacobian
matrix of $\mathbf{F}$ is given by
\begin{gather*}
\operatorname{Jac}:=\begin{pmatrix}
 x_1^2x_2^2 & 2x_0x_1x_2^2 & 2x_0x_1^2x_2 & 0& 0 \\
 -4 & -8 & -16 & 2y_0 & 8y_2 \\
 16x_2-4y_2^2 & \bigstar & -4y_0^2 & \blacklozenge &\vartriangle
 \end{pmatrix}\\
 \hphantom{\operatorname{Jac}:=}{}\text{with}
 \begin{cases}
 \bigstar=-4y_0y_2+16x_2+8x_1, \\
 \blacklozenge=2y_0\bigl(y_2^2-4x_2\bigr)-4x_1y_2, \\
 \vartriangle= 2y_2\bigl(y_0^2-4x_0\bigr)-4x_1y_0.
 \end{cases}
 \end{gather*}
 Let $p_0\bigl(1,1,1,1,-\frac{3}{2},1\bigr)$ be
a point of $\mathcal{H}$. The Jacobian matrix of $\mathbf{F}$ at
$p_0$ is given by
\[\operatorname{Jac}(p_0)=\begin{pmatrix}
 \hphantom{-}1 & \hphantom{-}2 & \hphantom{-}2 & \hphantom{-}0 & \hphantom{-}0 \\
 -4 & -8 & -16 & \hphantom{-}2 & \hphantom{-}8 \\
 \hphantom{-}12 & \hphantom{-}20 & \hphantom{-}28 & -10 & -10
 \end{pmatrix}.\]
The rank of matrix $\operatorname{Jac}(p_0)$ is $3$, so the
differentials ${\rm d}\mathbf{F}_i$, $i=1,2,3$, are
independent at $p_0$ and since the functions $\mathbf{F}_i$ are
polynomials, $\mathbf{F}$ is independent on a dense open subset
$\mathcal{U}_{\mathbf{F}}$ of $\mathcal{H}$.
We~can deduce that $(\mathcal{H},\{\cdot,\cdot\},\mathbf{F})$ is an integrable system in the sense of Liouville.

Let $\mathcal{S}$ be the set of points in $\mathcal{H}$ where the
determinants of all $3 \times 3$ minors of the matrix $\operatorname{Jac}$ cancel.
By direct computation, we prove that $S$ is the union of following
subvarieties:
\begin{gather*}
\mathcal{S}_1:= \{x_1=0\},\qquad \mathcal{S}_2:= \{x_2=0\},\qquad \mathcal{S}_3:=\big\{x_1=2y_2^2,\,x_0=4x_2,\,y_0=4y_2\big\},\\
\mathcal{S}_4:= \biggl\{y_0=y_2=0,\,x_1=\frac{2x_2^2+x_0x_2+ \alpha_1}{x_0},\,x_1=\frac{2x_2^2+x_0x_2-\alpha_1}{x_0}\biggr\},\\
\mathcal{S}_5:= \biggl\{x_0=\frac{1}{4}\beta,\, x_2=\frac{1}{2}\frac{y_2\bigl(4x_1y_2-2x_1y_0-2y_0y_2^2+y_0^2y_2\bigr)}{y_0(-4y_2+y_0)}\biggr\}
\end{gather*}
 with
\begin{gather*}
\alpha_1=\sqrt{4x_2^4+4x_2^3x_0-7x_0^2x_2^2+2x_2x_0^3},\\
\beta=\frac{8x_1y_0y_2-8x_1y_2^2+4y_0y_2^3-4y_0^2y_2^2-2x_1y_0^2+y_0^3y_2}{y_2(-4y_2+ y_0 )}.
\end{gather*}

 The images under $\mathbf{F}$ of $S_1$ and $S_2$ are contained in the subset $ F_1=c_1=0.$
 By substituting $x_1=\frac{1}{8}y_0^2$, $x_2=\frac{1}{4}x_0$ and $y_0=4y_2$ in the three constants
of motion $F_i=c_i$, $i=1,2,3$, by direct computation
with \textsc{Maple}, we obtain
\[
 c_1=\dfrac{1}{4}x_0^3y_2^4, \qquad
 c_2=4y_2^2-8x_0, \qquad
 c_3=4x_0\bigl(-3y_2^2+x_0\bigr).
\]
Using $F_i=c_i$, $i=1,2,3$, and eliminating the variables
$x_i$, $y_i$, $i=1,2,3$, in the above expresions, we deduce that
the image under $F$ of $S_3$ is contained in the subset
\[256\bigl(3200000c_1^2+2000c_3^2c_2c_1-225c_3c_2^3c_1
+c_3^5\bigr)+1728c_2^5c_1 -32c_3^4c_2^2+c_3^3c_2^4=0,\] and the set
of regular values of the momentum map $\mathbf{F}$ is the Zariski
open subset $\Omega$ defined by
\begin{gather*}
 \Omega = \big\{c=(c_1,c_2,c_3) \in \mathbb{C}^3 \mid c_1\neq 0 \ \text{and} \\ \hphantom{\Omega = \big\{}{}
 256\bigl(3200000c_1^2+2000c_3 ^2c_2c_1-225c_3c_2^3c_1+c_3^5\bigr)+1728c_2^5c_1-32c_3^4c_2^2+c_3^3c_2^4\neq 0\big\}.
\end{gather*}
At a generic point $c=(c_1,c_2,c_3) \in \mathbb{C}^3$,
the fiber on $c\in \Omega$ of $\mathbf{F }$ is therefore
\[\mathbb{F}_c:= \mathbb{F}^{-1}(c)=\bigcap_{i=1}^3
\{m\in \mathcal{H}\mid \mathbf{F}_i(m)=c_i\}.\]
Hence, we have the following result which prove that
\smash{$a_4^{(2)}$} Toda lattice is a completely integrable
system in the Liouville sense.

\begin{Proposition}
For $c\in \Omega$, the fiber $\mathbb{F}_c$ over $c$ of the momentum
$F$ is a smooth affine variety of dimension $2$ and the rank of the
Poisson structure \eqref{matrij2} is maximal and equal to $4$ at
each point of $\mathbb{F}_c$; moreover, the vector fields
$\mathcal{V}_1$ and $\mathcal{V}_2$ are independent at each point
of the fiber $\mathbb{F}_c$.
\end{Proposition}

\begin{Proposition}
$(\mathcal{H}, \{\cdot,\cdot\},\mathbf{F})$ is a
completely integrable system
 describing the $a_4^{(2)}$ Toda lattice, where
$\mathbf{F}=(F_1,F_2,F_3)$ and $\{\cdot,\cdot\}$ are
given respectively by \eqref{invar1a4} and \eqref{matrij2} with
commuting vector fields \eqref{systa4} and \eqref{syst3}.
\end{Proposition}
\subsection[Algebraic integrability of a\_4\^{}(2) Toda lattice]{Algebraic integrability of $\boldsymbol{a_4^{(2)}}$ Toda lattice}

To show that the $a_4^{(2)}$ Toda lattice is
algebraically completely integrable, we show that for all
$c=(c_1,c_2,c_3)\in \Omega$, the fiber $\mathbf{F_c}$ is
the affine part of an abelian surface on which the two vector fields
$\mathcal{V}_1$ and $\mathcal{V}_2$ are linear \cite{Dehainsala}.
For that it is necessary to verify that $\mathbf{F_c}$
satisfies the conditions of the Liouville complex theorem (Theorem~\ref{liouville}).

 Observe that the vector field $\mathcal{V}_1$ is homogeneous with respect to the following weights:
\[\varpi(x_0,x_1,x_2,y_0,y_1,y_2)=(2,2,2,1,1,1).\]
We can also verify that the constants of motion $F_i$ are weight
homogeneous and they have the following weights:
$\varpi(F_1,F_2,F_3)=(10,2,4).$
\subsubsection{Laurent solutions}

According to \cite{al_Moer3}, the vector fields of an algebraic complete integrable system
have good properties at infinity; after all, a linear vector field
on a complex torus is the same along the divisor, which will happen
to be absent in phase space, as on the rest of the torus. In fact,
since every holomorphic function on a complex torus can be written
as a quotient of theta functions, the integral curves (solutions) to
any of the vector fields of an algebraic complete integrable system can be written as a
quotient of holomorphic functions. Intuitively speaking this means
that we can consider not only Taylor solutions to the differential
equations that describe these vector fields but also Laurent
solutions, which will correspond to initial conditions at infinity (precisely: on the divisor that needs to be adjoined to the fibers of
the momentum map to complete them into abelian varieties). Moreover,
these Laurent solutions must depend on $\dim \mathcal{H}-1$
free parameters, which corresponds to the freedom of choice of the
initial condition at infinity.

Let us look for weights homogeneous Laurent solutions associated to
the vector field $\mathcal{V}_1$. The~solution form of vector field
$\mathcal{V}_1$ is
\[x_i(t)=\frac{1}{t^{2}}\sum_{k=0}^\infty x_i
^{(k)}t^k \qquad \text{and} \qquad
y_i(t)=\frac{1}{t}\sum_{k=0}^\infty
y_i^{(k)}t^k, \qquad i=0,1,2,\]
that is to say
that $\varpi(x_i)=2\varpi(y_i)=2$ for $i=0,1,
$2. By substituting these solutions into the differential equations
\eqref{systa4} associated with the vector field $\mathcal{V}_1$, the
indicial locus is the subset of~$\mathcal{H}$ given according to
\eqref{equaindicial} by
\begin{gather*}
 0 = x_0^{(0)}\bigl(2+y_0^{(0)}\bigr), \\
 0 = x_1^{(0)}\bigl(2+y_1^{(0)}\bigr),\\
 0 = x_2^{(0)}\bigl(2+y_2^{(0)}\bigr), \\
 0 = y_0^{(0)}+2x_0^{(0)}-2x_1^{(0)},\\
 0 = y_1^{(0)}-x_0^{(0)}+2x_1^{(0)}-2x_2^{(0) } ,\\
 0 = y_2^{(0)}-x_1^{(0)}+2x_2^{(0)}.
\end{gather*}
These equations are easily solved and they yield the following
(non-zero) solutions:
 \begin{alignat*}{3}
& m_0= (1,0,0,-2,1,0), \qquad && m_3= (0,4,3,8,-2,-2), &\\
& m_1= (0,1,0,2,-2,1), \qquad && m_4= (1,0,1,-2,3,-2), &\\
& m_2= (0,0,1,0,2,-2), \qquad && m_5= (4,3,0,-2,-2,3).&
 \end{alignat*}
 The Kowalevski matrix at an arbitrary point
$\bigl(x^{(0)},y^{(0)}\bigr)$,
 solution of the indicial equation is given by
 \[\mathcal{K}\bigl(x^{(0)},y^{(0)}\bigr)
 =\begin{pmatrix}
 2+ y_0^{(0)}& 0 & 0 & x_0^{(0)} & 0 \\
 0 & 2- \frac{1}{2}y_0^{(0)}-y_2^{(0)}& 0 &-\frac{1}{2}x_1^{ (0)} & -x_1^{(0)} \\[1mm]
 0 & 0 & 2+ y_2^{(0)} & 0 & x_2^{(0)}\\
 2 & -2 & 0 & 1 & 0 \\
 0 & -1 & 2 & 0 & 1 \\
 \end{pmatrix}.\]

 \begin{Lemma}
The system of differential equation \eqref{systa4} of the vector
field $\mathcal{V}_1$ has three distinct families of homogeneous
Laurent solutions with weights depending on four $(\dim\mathcal{H}-1)$ free parameters.
\end{Lemma}
\begin{proof}

The characteristic polynomial of the Kowalevski matrix at $m_0$ is
given by
 \[\mathcal{X}( \lambda;m_0)= ( \lambda-1 ) ( \lambda-3 ) ( \lambda+1
 )( \lambda-2 ) ^{2}.\]

 It admits $4$ non-negative eigenvalues which leads
to a Laurent solution depending on four free parameters, whose the
five leading terms (going with steps $1$, $2$, $2$, $3$, respectively, are
denoted by $a$, $c$, $d$ and $e$). The first four terms of this balance
which we note $x(t,m_0)$:
\begin{gather}
 x_0(t,m_0) = \frac{1}{t^2}+c-\frac{1}{2}ed+O\bigl(t^2\bigr), \nonumber\\
 x_1(t,m_0) = et+O\bigl(t^2\bigr), \nonumber\\
 x_2(t,m_0) = d-adt+O\bigl(t^2\bigr), \nonumber\\
 y_0(t,m_0) = -\frac{2}{t}+2ct-\frac{3}{2}et^2+O\bigl(t^3\bigr), \nonumber\\
 y_1(t,m_0) = \frac{1}{t}+a-(c+2d)t+\biggl(ad+\frac{5}{4}e \biggr)t^2+O\bigl(t^3\bigr),\nonumber\\
 y_2(t,m_0) = -a+2dt-\biggl(ad+\frac{1}{2}e\biggr)t^2+O\bigl(t^3 \bigr).\label{balancemainm0}
\end{gather}

The characteristic polynomial of the Kowalevski matrix at $m_1$
 is given by
 \[\mathcal{X}( \lambda;m_1)= ( \lambda-4 ) ( \lambda-1 ) (
 \lambda-2 )( \lambda-3 )( \lambda+1 ).
\]
It admits 4 non-negative eigenvalues which leads
to a Laurent solution depending on four free parameters, whose the
five leading terms (going with steps $1$, $2$, $3$, $4$, respectively, are
denoted by $a$, $c$, $d$ and $e$). The first four terms of this balance
which we note $x(t,m_1)$:
\begin{gather}
 x_0(t,m_1) = et^2 +O\bigl(t^3\bigr), \nonumber\\
 x_1(t,m_1) = \frac{1}{t^2}+c-\frac{1}{2}d t+\frac{1}{10}\bigl( 6c^2+ad-e\bigr)t^2+O\bigl(t^3\bigr), \nonumber\\
 x_2(t,m_1) = dt-\frac{1}{2}a d t^2+O\bigl(t^3\bigr), \nonumber\\
 y_0(t,m_1) = \frac{2}{t}+a-2ct+\frac{1}{2}d t^2+\frac{1}{15} \bigl(11e-6c^2-ad\bigr)t^3+O\bigl(t^4\bigr), \nonumber\\
 y_1(t,m_1) = -\frac{2}{t}+2ct-\frac{3}{2}d t^2+\frac{2}{5} \bigl(c^2+ad-e\bigr)t^3+O\bigl(t^4\bigr), \nonumber\\
 y_2(t,m_1) = \frac{1}{t}-\frac{1}{2}a-c t+\frac{5}{4}dt^2+\frac{ 1}{30}\bigl(e-6c^2-11ad\bigr)t^3+O\bigl(t^4\bigr). \label{balancemainem1}
\end{gather}
The characteristic polynomial of the Kowalevski matrix at $m_2$
 is given by
 \[\mathcal{X}( \lambda;m_2)= ( \lambda+1 ) ( \lambda-4 ) ( \lambda-1) ( \lambda-2 ) ^{2}.
\]It admits 4 non-negative eigenvalues which leads
to a Laurent solution depending on four free parameters, whose the
five leading terms (going with steps $1$, $2$, $2$, $4$, respectively, are
denoted by $a$, $c$, $d$ and $e$). The first four terms of this balance
which we note $x(t,m_2)$:
\begin{gather}
 x_0(t,m_2) = c+act+\frac{1}{2}c\bigl(2c+a^2\bigr)t^2+O\bigl(t^3 \bigr), \nonumber\\
 x_1(t,m_2) = et^2+O\bigl(t^3\bigr), \nonumber\\
 x_2(t,m_2) = \frac{1}{t^2}+d+\frac{1}{10}\bigl(6d^2-e\bigr)t^2 +O\bigl(t^3\bigr), \nonumber\\
 y_0(t,m_2) = a+2ct+act^2+\frac{1}{3}\bigl(2c^2+a^2c-2e\bigr)t^3 +O\bigl(t^4\bigr), \nonumber\\
 y_1(t,m_2) = \frac{2}{t}-\frac{1}{2}a-(c+2d)t-\frac{1 }{2}act^2-\frac{1}{30}\bigl(10c^2+5a^2c+12d^2-22e\bigr)t^3+O\bigl(t^4\bigr), \label{balancemainem2}\\
 y_2(t,m_2) = -\frac{2}{t}+2dt+\frac{2}{5}\bigl(d^2-e\bigr)t^3+ O\bigl(t^4\bigr).
\tag*{\qed}
\end{gather}
\renewcommand{\qed}{}
\end{proof}

According to Adler and van Moerbeke \cite{al_Moer4}, since the Toda
problem is algebraic integrable, \cite[Theorem~1]{al_Moer2}~implies the existence of a coherent tree of Laurent solutions,
satisfying all the conditions described in \cite{al_Moer3}. The
existence of a~coherent tree of Laurent solutions is crucial for
obtaining the complete structure of the divisor $\mathcal{D}$ at
infinity in \cite{al_Moer4}. In this work, we do not use a~coherent
tree of Laurent solutions, but only the so-called principal balances
depending on $4$~free parameters corresponding to $m_0$, $m_1$, $m_2$
and that $m_3$, $m_4$, $m_5$ correspond to lower balances depending on
$3$ free parameters.

\subsubsection[Painlev\'e divisors of a\_4\^{}(2) Toda lattice]{Painlev\'e divisors of $\boldsymbol{a_4^{(2)}}$ Toda lattice}

We now search the formal Painlev\'e divisors, i.e., the algebraic
curves defined by the three different principal balances
$x(t,m_i)_{i=0,1,2}$, confined to a fixed affine invariant surface
$\mathbb{F}_c$, $c \in \Omega$. We have the following assertion.

\begin{Proposition}\label{proposition3.4}
For the Laurent solution $x(t;m_0)$ restricted to the
invariant surface, for $c \in \Omega$, the Painlev\'e divisor
\smash{$\Gamma_c^{(0)}$} is a smooth genus three hyperelliptic
curve. It is given by
\begin{gather*}
 \Gamma_c^{(0)}\colon \ 16d^2a^8-\bigl(256d^3+8d^2c_2\bigr)a^6+ \bigl(1536d^2+96dc_2+8c_3+c_2^2\bigr)d^2a^4 \\
 \hphantom{\Gamma_c^{(0)}\colon} \ {}-\bigl(\bigl(8\bigl(8c_3+48dc_2+c_2^2+512d^2\bigr)d+2c_2c_3\bigr)d^2+64c_1\bigr)a^2 \\
 \hphantom{\Gamma_c^{(0)}\colon} \ {} +\bigl(8d\bigl(c_2c_3+16dc_3+64d^2c_2 +512d^3+2dc_2^2\bigr)+c_3^2\bigr)d^2=0.
\end{gather*}
It is completed in a Riemann surface, denoted
\smash{$\overline{\Gamma}_c^{(0)}$} by adding $8$ points to
infinity.
\end{Proposition}

\begin{proof}
 Consider Laurent's solution $x(t;m_0)$ in \eqref{balancemainm0}.
 When it is substituted in the equations $F_i=c_i$, $i=1,2,3$, where $c_i=(c_1,c_2,c_3)\in \Omega$, we get
\begin{gather*}
 c_1 = e^2d^2, \qquad
 c_2 = 4a^2-12c-16d, \qquad
 c_3 = -12ca^2+48cd-8ae.
\end{gather*}
By eliminating the parameters $c$ and $e$, we obtain an algebraic
relation between $a$ and $d$ which is the equation of an affine
curve in $\mathbb{C}^2$ defined by
\begin{gather*}
\begin{split}
& \Gamma_c^{(0)}\colon \ 16d^2a^8-\bigl(256d^3+8d^2c_2\bigr)a^6+\bigl(1536d^2+96dc_2+8c_3+c_2^2\bigr)d^2a^4\\
&\hphantom{\Gamma_c^{(0)}\colon} \ {} -\bigl(\bigl(8\bigl(8c_3+48dc_2+c_2^2+512d^2\bigr)d+2c_2c_3\bigr)d^2+64c_1\bigr)a^2\\
& \hphantom{\Gamma_c^{(0)}\colon} \ {} +\bigl(8d\bigl(c_2c_3+16dc_3+64d^2c_2 +512d^3+2dc_2^2\bigr)+c_3^2\bigr)d^2=0.
\end{split}
\end{gather*}

The affine curve \smash{$ \Gamma_c^{(0)}$} is smooth for $c\in
\Omega$. It is completed in a Riemann surface, denoted~\smash{$\overline{\Gamma}_c^{(0)}$} by adding $8$ points to
infinity denoted by $\infty_{\epsilon}$,
$\infty_{\epsilon_1\epsilon_2}$ and $\infty^3_{\epsilon_3}$ with
$\epsilon^2=\epsilon_1^2=\epsilon_2^2=\epsilon_3^2=1$ and
$\delta=\sqrt{c_2^2-16c_3}$. A neighborhood of each of these points
is described according to a local parameter $\varsigma$ by
\begin{alignat}{4}
& \infty_{\epsilon}\colon \ && a =\varsigma^{-1},\qquad && d
=\frac{1}{2}\epsilon c_2\sqrt{c_1}\varsigma^5+
2\epsilon\sqrt{c_1}\varsigma^3 + O\bigl(\varsigma^5\bigr), &\label{valeurde_inftyvarepsilon_D0}\\
& \infty_{\epsilon_1\epsilon_2}\colon \ && a =\varsigma^{-1},\qquad && d
=\frac{\frac{1}{4}+\bigl(-\frac{1}{32}c_2+\frac{1}{32}\epsilon_1\delta\bigr)
+8\epsilon_2\varsigma^3\sqrt{\frac{c_1}{c_2^2-16c_3}}}{\varsigma^2}+O(\varsigma), &\label{valeurde_inftyvarepsilon1epsilon2_D0}\\
& \infty^3_{\epsilon_3}\colon \ && a =\varsigma,\qquad && d
=-\frac{256c_1c_2\varsigma^2}{c_3^3}+8\epsilon_3\varsigma\sqrt{\frac{c_1}{c_3^2}}+O\bigl(\varsigma^3\bigr), &\label{valeurde_inftyvarepsilon3_D0}
\end{alignat}
This completes the proof of Proposition~\ref{proposition3.4}.
\end{proof}

\begin{Proposition}\label{proposition3.5}
For the Laurent solution $x(t;m_1)$ restricted to the
invariant surface, for $c \in \Omega$, the Painlev\'e divisor
\smash{$\Gamma_c^{(1)}$} is a smooth genus four curve. It is
given by
\begin{equation*}
 \Gamma_c^{(1)}\colon \ 256ad^3-\bigl(\bigl(4a^2-c_2\bigr)^2-16c_3\bigr)d^2+64c_1=0.
\end{equation*}
It is completed in a Riemann surface, denoted
\smash{$\overline{\Gamma}_c^{(1)}$} by adding $4$ points to
infinity.
\end{Proposition}

\begin{proof}
Consider Laurent's solution $x(t;m_1)$ in \eqref{balancemainem1}. When it was substituted in the
equations $F_i=c_i$, $i=1,2,3$, where
$c_i=(c_1,c_2,c_3)\in \Omega$, we
obtain
\begin{gather*}
 c_1 = ed^2, \qquad
 c_2 = 2a^2-24c, \qquad
 c_3 = \frac{1}{4}a^4+6a^2c-16ad+36c^2-4e
\end{gather*}
by eliminating the parameters $c$ and $e$, we obtain an algebraic
relation between $a$ and $d$ which is the equation of an affine
curve in $\mathbb{C}^2$ defined by
\begin{equation*}
 \Gamma_c^{(1)}\colon \ 256ad^3-\bigl(\bigl(4a^2-c_2\bigr)^2-16c_3\bigr)d^2+64c_1=0.
\end{equation*}

The affine curve $ \Gamma_c^{(1)}$ is smooth for $c\in
\Omega$. Indeed, let
\[g(a,d)=256ad^3-\bigl(\bigl(4a^2-c_2\bigr)^2-16c_3\bigr)d^2+64c_1, \]
we have
\begin{gather*}
 \frac{\partial g}{\partial a}(a,d) = 16d^2\bigl(4a^3-ac_2-16d\bigr),\\
 \frac{\partial g}{\partial d}(a,d) = 2d\bigl(16a^4-8a^2c_2-384ad+c_2^2-16c_3\bigr).
\end{gather*}
Thus, a point $(a,d)$ is singular for the affine curve $
\Gamma_c^{(1)}$ if \[g(a,d)= \frac{\partial
g}{\partial a}(a,d)= \frac{\partial g}{\partial
d}(a,d)=0\] as $d\neq 0$ since $c_1
\neq 0$, then
\begin{gather*}
 \begin{cases}
 4a^3-ac_2-16d=0, \\
 16a^4-8a^2c_2-384ad+c_2^2-16c_3=0 ,\\
 -256ad^3+\bigl(16a^4-8a^2c_2-384ad+c_2^2-16c_3\bigr)d^2-64c_1=0,
 \end{cases} \\
\qquad{}\Leftrightarrow
 \begin{cases}
 4a^3-ac_2-16d=0, \\
 16a^4-8a^2c_2-384ad+c_2^2-16c_3=0 ,\\
 2ad^3-c_1=0.
 \end{cases}
\end{gather*}

By expressing $a$ as a function of $d$ in the last equation of the
system and substituting the expression into both first equations, we
obtain
\begin{gather*}
a=\frac{c_1}{2d^3}, \qquad -c_1^3+c_1c_2d^6+32d^{10}=0, \\
c_1-2c_1^2c_2d^6+c_2^2d^{12}-192c_1d^{10}-16c_3d^{12}=0.
\end{gather*}
 The resultant of these two last polynomials in $d$ is given by the following polynomial in terms of $c_1$,
$c_2$ and $c_3$ up to a constant
\[c_1^{36}\bigl(819200000c_1^2+512000c_3^2c_2c_1-57600c_3c_2^3c_1-32c_3^4c_2^2+c_3^3c_2^4+256c_3^5+1728c_2^5c_1\bigr)^2.
\]

This expression is not zero for $c \in \Omega$. We deduce that
\smash{$\Gamma_c^{(1)}$} is a smooth affine curve for $c \in
\Omega$. It is complete into a Riemann surface, denoted
\smash{$\overline{\Gamma}_c^{(1)}$} by adding four points at
infinity denoted by $\infty^{1}$, $\infty^{2}$ and
$\infty_{\epsilon}$. A neighborhood of each of these points is
described according to a local parameter $\varsigma$ by
\begin{alignat}{4}
&\infty_{\epsilon}\colon \ && d =\varepsilon\sqrt{4}\sqrt{c_1}\varsigma^2,\qquad&&
a =\frac{\bigl( \frac { c_2^2}{128}-\frac{c_3}{4} \bigr)\varsigma^5 +\frac{c_2\varsigma^3}{8} +\varsigma}{\varsigma^2} +O\bigl(\varsigma^3\bigr),&\label{valeurde_inftyvarepsilon_D1}\\
&\infty^1\colon \ && d =\varsigma^{-1},\qquad &&
a=\biggl( \frac {c_2^{2}}{256} -\frac{c_3}{16}\biggr)\varsigma+O\bigl(\varsigma^2\bigr),&\label{valeurde_infty1_D1}\\
&\infty^2\colon \ && d =\frac{1}{16\varsigma^3},\qquad&&
a=\frac{\frac{1}{3} \bigl(16c_3-c_2^2\bigr)\varsigma^{6} +\frac{8c_2}{3}\varsigma^{4}+16\varsigma^{2}}{16{\varsigma}^{3}} +O\bigl(\varsigma^3\bigr),&\label{valeurde_infty2_D1}
\end{alignat}
This completes the proof of Proposition~\ref{proposition3.5}.
\end{proof}

\begin{Proposition}
For the Laurent solution $x(t;m_2)$ restricted to the
invariant surface, for $c \in \Omega$, the Painlev\'e divisor
\smash{$\Gamma_c^{(2)}$} is a smooth genus two hyperelliptic
curve. It is given by
\begin{equation*}
 \Gamma_c^{(2)}\colon \ e^4a^4-\bigl(8c_1+c_2e^2\bigr)a^2e^2-64e^5+4e^2c_1c_2+4c_3e^4+16c_1^2=0.
\end{equation*}
It is completed in a Riemann surface, which is a double covering of
$\mathbb{P}^1$ ramified into $5$ points.
\end{Proposition}

\begin{proof}
For a point $c=(c_1,c_2,c_3) \in \mathbb{C}^3$, by
substituting the Laurent's solution $x(t;m_2)$ of
\eqref{balancemainem2} inside equations $F_i=c_i$, $i=1,2,3$, we find
the independent algebraic expressions of $t$, namely the three
algebraic relations between the parameters $a$, $c$, $d$ and $e$ below
\begin{eqnarray*}
 c_1 = ce^2, \qquad
 c_2 = a^2-4c-48d, \qquad
 c_3 = -12a^2d+48cd+16e.
\end{eqnarray*}
For $c\in \Omega$, $c_1\neq0$ therefore the parameters $c$, $d$ and $e$
are not zero. The first two equations are linear in the parameters
$c$ and $d$, and can be solved linearly in these parameters as a
function of the constants of motion thus giving
\begin{equation*}
 c=\frac{c_1}{e^2} \qquad \text{and} \qquad
 d=\frac{1}{48}\biggl(a^2-c_2-\frac{4c_1}{e^2}\biggr).
\end{equation*}
The third equation then reduces to the following equation of an
affine curve in $\mathbb{C}^2$.
\begin{equation}\label{courbeaffine}
 \Gamma_c^{(2)}\colon \ e^4a^4-\bigl(8c_1+c_2e^2\bigr)a^2e^2-64e^5+4e^2c_1c_2+4c_3e^4+16c_1^2=0.
\end{equation}

The affine curve $ \Gamma_c^{(2)}$ is smooth for $c\in
\Omega$. In fact, let
\[f(a,e)=e^4a^4-\bigl(8c_1+c_2e^2\bigr)a^2e^2-64e^5+4e^2c_1c_2+4c_3e^4+16c_1^2, \]
we have
\begin{gather*}
 \frac{\partial f}{\partial a}(a,e) = 4a^3e^4+2a\bigl(-8c_1e^2-c_2e^4\bigr),\\
 \frac{\partial f}{\partial e}(a,e) = 4a^4e^3+\bigl(-16c_1e-4c_2e^3\bigr)a^2-320e^4+8c_1c_2e+16c_3e^3.
\end{gather*}
So, a point $(a,e)$ is singular for the affine curve \smash{$\Gamma_c^{(2)}$} if \smash{$f(a,e)= \frac{\partial
f}{\partial a}(a,e)= \frac{\partial f}{\partial
e}(a,e)=0$} as $e\neq 0$ so either $a=0$,
 or \smash{$a^2=\frac{8c_1+c_2e^2}{2e^2}$}.
\begin{itemize}\itemsep=0pt
 \item If $a^2=\frac{8c_1+c_2e^2}{2e^2}$, after substitution in
 the equations $f(a,e)=0 \text{ and } \frac{\partial f}{\partial
 e}(a,e)=0$, we then obtain
 \[-\frac{1}{4}e^4\bigl(c_2^2+256e-16c_3\bigr)=0 \qquad \text{and} \qquad -e^3\bigl(c_2^2+320e-16c_3 \bigr)=0.\]
As $e\neq0$, we have $c_2^2-16c_3=-256e=-320e$, which implies that
$e=0$, which is absurd!
 \item If $a=0$, this leads to the following system:
\begin{gather*}
 -64e^5+4c_1 c_2e^2+4c_3e^4+16c_1^2=0 ,\\
 -320e^4+8c_1c_2e+16c_3e^3=0.
 \end{gather*}
 The resultant of these two polynomials in $e$ composing the
system is given up to a constant by the following polynomial in
terms of $c_1$, $c_2$ and $c_3$:
\[c_1\bigl(819200000c_1^2+512000c_3^2c_2c_1-57600c_3c_2^3c_1-32c_3^4c_2^2+c_3^3c_2^4+256c_3^5+1728c_2^5c_1\bigr).
\]
\end{itemize}
This expression is not zero for $c \in \Omega$. We deduce that
\smash{$\Gamma_c^{(2)}$} is a smooth affine curve for $c \in
\Omega$.

The equation \eqref{courbeaffine} of the affine curve
\smash{$\Gamma_c^{(2)}$} can be written as
\[\biggl(a^2-\frac{4c_1}{e^2}\biggr)\biggl(a^2-\frac{4c_1}{e^2}-c_2\biggr)=64e-4c_3.\]
 We deduce that \smash{$\Gamma_c^{(2)}$} is a double cover of the rational affine curve
 \[\varepsilon_c^{(2)}\colon \ u(u-c_2)-64e+4c_3=0,\]
the map which links the two curves is explicitly given by
\begin{align*}
 \psi\colon \ \Gamma_c^{(2)} \longrightarrow \varepsilon_c^{(2)}, \qquad
 (a,e) \longmapsto (u,e)=\biggl(a^2-\frac{4c_1}{e^2},e\biggr).
\end{align*}
$\varepsilon_c^{(2)}$ being irreducible curve of genus
$0$, a parametrization of $\varepsilon_c^{(2)}$ is given
by \[ \varepsilon_c^{(2)}=\biggl\{
(u,e)=\biggl(t,\frac{1}{64}\bigl(t^2-c_2t+4c_3\bigr)\biggr),\, t
\in \mathbb{C}\biggr\}. \]
 If $\psi(a,e)=(u,e)$, then we have
 \begin{equation*}
 a=\pm\sqrt{\frac{t^5-2c_2t^4+\bigl(8c_3+c_2^2\bigr)t^3-8c_2c_3t^2+16c_3^2t+16384c_1}{\bigl(t^2 -c_2t+4c_3\bigr)^2}},
\qquad e=\frac{1}{64}\bigl(t^2-c_2t+4c_3\bigr).
\end{equation*}
The branch points of the cover \smash{$\psi\colon \Gamma_c^{(2)}
\longrightarrow \varepsilon_c^{(2)}$} are the points
$(u,e)$ of the curve \smash{$\varepsilon_c^{(2)}$} for
which $a =0$, i.e., where $u$ is a root of the polynomial of degree
$5$ given by
 \[P(t)=t^5-2c_2t^4+\bigl(8c_3+c_2^2\bigr)t^3-8c_2c_3t^2+16c_3^2t+16384c_1,\]
 we verify that the discriminant of $P(t)$ is given, up to a constant, by
\[c_1^2\bigl(819200000c_1^2+512000c_1c_2c_3^2-57600c_1c_2^3c_3+1728c_1c_2^5+256c_3^5-32c_2^2c_3^4+c_2^4c_3^3\bigr),\]
which for $c \in \Omega$ is not zero. Therefore, these five branch points are distinct. Thus, the map $\psi$ admits five ramification points on \smash{$\Gamma_c^{(2)}$}.

The curve \smash{$\Gamma_c^{(2)}$} can be completed into a compact Riemann surface, denoted \smash{$\overline{\Gamma}_c^{(2)}$} by adding to it
 five points to infinity denoted by $\infty$, $\infty_{\epsilon_1\epsilon_2}$
 where $\epsilon_1^2=\epsilon_2^2=1$ and \smash{$\delta=\sqrt{c_2^2-16c_3}$}. Neighborhoods of these points are described according to a local parameter $\varsigma$ by
\begin{alignat}{4}
&\infty_{\epsilon_1\epsilon_2}\colon \ && a =\varsigma^{-1}, \qquad &&
e =2\epsilon_1\sqrt{c_1}\varsigma\biggl(1+\frac{1}{4}(c_2+\epsilon_2\delta)\varsigma^2+O\bigl(\varsigma^4\bigr)\biggr), &\label{valeurde_inftyvarepsilonm2}\\
&\infty\colon \ && a =\varsigma^{-1}, \qquad &&
e=\frac{1}{64}\bigl(\varsigma^{-4}-c_2\varsigma^{-2}+4c_3+O\bigl(\varsigma^6\bigr)\bigr).&\label{valeurde_inftym2}
\end{alignat}

 When the application $\psi$ is extended in application
 \smash{$\overline{\psi}\colon \overline{\Gamma}_c^{(2)} \longrightarrow
 \overline{\varepsilon}_c^{(2)}$}, there is another branching point. Indeed, if we write $t=\frac{1}{\varsigma^2}$ depending on a local parameter $\varsigma$, we obtain the point at infinity \smash{$\infty \in
 \overline{\Gamma}_c^{(2)}\backslash\Gamma_c^{(2)}$} given by
 \[(a,e)=\biggl(\varsigma^{-1},\frac{1}{64}\bigl(\varsigma^{-4}-
 c_2\varsigma^{-2}+4c_3+O\bigl(\varsigma^6\bigr)\bigr)\biggr).\]
 There is no other branching point. Indeed, noting
 $t=t_1+\varsigma^2$ and $t=t_2+\varsigma^2$ local parameterizations in the neighborhood respectively of $t_1$ and $t_2$ the roots of $t^2-c_2t+4c_3$, it follows that
 \[a=\pm\frac{128}{\varsigma}\sqrt{\frac{c_1} c_2^2-16c_3}+O (1 ),\] then respectively
 \[e=\frac{1}{64}c_2 (c_2+\delta )+O \bigl(\varsigma^2 \bigr) \qquad \text{ and } \qquad
 e= \frac{1}{64}c_2 (c_2-\delta )+O \bigl(\varsigma^2 \bigr),
 \] which shows that above from the points $t_1$ and $t_2$, the map $\overline{\psi}$ is not ramified. We then conclude that the application $\overline{\psi}\colon \overline{\Gamma}_c^{(2)} \longrightarrow
 \overline{\varepsilon}_c^{(2)}$ is a double covering of $\mathbb{P}^1$ branched into $5$ points. We deduce that the genus of $\overline{\Gamma}_c^{(2)}$ is equal to $2$ according to the Riemann--Hurwitz formula.
 \end{proof}

\begin{Remark}
The above propositions only compute an affine part of the divisor
$\mathcal{D}$, since the lower balances depending on $3$ free
parameters are not used. The completed non singular curves of
respective genus $3$, $4$ and $2$ mentioned in these propositions
correspond to blowing up the singular points of the three
irreducible components of $\mathcal{D}$.
\end{Remark}
 \subsection{Abelian surface}
 According to \cite{Dehainsala1}, in order to embed the three Riemann surfaces \smash{$ \overline{\Gamma}_c^{ (0 )}$}, \smash{$\overline{\Gamma}_c^{ (1 )}$} and \smash{$ \overline{\Gamma}_c^{(2)}$} into some projective space, one of the key underlying
principles used the Kodaira embedding theorem, which states that a
smooth complex manifold can be smoothly embedded into projective
space~$\mathbb{P}^{N} (\mathbb{C} )$ with the set of
functions having a pole of order $k$ along positive divisor on the
manifold, provided~$k$ is large enough; fortunately, for abelian
surfaces, $k$ need not be larger than three according to Lefschetz theorem. These functions are easily constructed from the three
Laurent solutions by looking for polynomials in the phase variables
which in the expansions have at most a $k$-fold pole. The nature of
the expansions and some algebraic properties of abelian varieties
provide a recipe for when to terminate our search for such
functions, thus making the procedure implementable. Precisely, we
wish to find a set of polynomial functions $\{z_0, \dots, z_N\}$,
of increasing degree in the original variables $x_0, \dots, y_2$
having the property that the embedding $\mathcal{D}$ of
\smash{$\overline{\Gamma}_c^{(i)}$}, $i=0,1,2$, into
$\mathbb{P}^{N} (\mathbb{C} )$ via those functions
satisfies the relation: $g (\mathcal{D} ) = N + 2$ where
$g (\mathcal{D} )$ is the arithmetic genus of~$\mathcal{D}$.
If the \smash{$a_4^{(2)}$} Toda lattice is an
irreducible algebraic complete integrable system, then as we have seen above a divisor
$\mathcal{D}$ can be added to a Zariski open subset of
$\mathcal{H}$, having the effect of compacting all fibers
$\mathbb{F}_c$, where $c \in \Omega$. The divisor that is added to
$\mathbb{F}_c$ will be denoted by $\mathcal{D}_c$ and the resulting
torus by $\mathbb{T}^2_c$. The vector fields $\mathcal{V}_1$ and
$\mathcal{V}_2$ extend to linear (hence holomorphic) vector fields
on this partial compactification of $\mathcal{H}$, hence we may
consider the integral curves of $\mathcal{V}_1$, starting from any
component \smash{$\mathcal{D}_c^{(i)}$}. Since the third power of an ample
divisor on an abelian variety is very ample, we look for all
polynomials which have a simple pole at most when any of the three
principal balances are substituted in them. Precisely, we~look for a
maximal independent set of functions which are independent when
restricted to $\mathbb{F}_c$. By direct computation, we obtain
twenty-five $(25)$ weight homogeneous polynomials of weight at most
$13$.

 These twenty-five ($25$) functions are defined by
\begin{gather}
\begin{aligned}
 &z_0= 1,\qquad && z_{11}= x_1x_2 (y_0-2y_2 ),&\\
& z_1= y_0,\qquad && z_{12}= -x_1x_2z_4,&\\
& z_2= y_2,\qquad && z_{13}= x_0x_1 \bigl(x_1-y_2^2 \bigr),& \\
& z_3= y_2^2 -x_1-4x_2,\qquad && z_{14}= x_0x_1x_2,&\\
& z_4= y_0y_2-2x_1,\qquad && z_{15}= x_0x_1x_2y_2 ,&\\
 & z_5=y_2 (y_0y_2-2x_1 )-4x_2y_0,\qquad && z_{16}= x_0x_1 (-y_0z_4+4x_0y_2 ),&\\
 & z_6= 2y_2 \bigl(y_2^2-4x_2 \bigr)+x_1 (y_0-4y_2 ),\qquad && z_{17}= x_1x_2 ( (y_0-2y_2 )z_4+8x_2y_0 ),& \\
& z_7= x_0x_1,\qquad && z_{18}= x_0x_1x_2 \bigl(4x_2-y_2^2 \bigr),&\\
 &z_8= x_1x_2,\qquad && z_{19}=x_0x_1x_2 \bigl(4x_0-y_0^2 \bigr),&\\
& z_9= -y_0y_2z_3+2x_1 \bigl(y_2^2-x_1 \bigr),\qquad && z_{20}=-x_0x_1x_2 ( -y_0z_4+4x_0y_2 ),&\\
& z_{10}= x_0x_1y_2,\qquad && z_{21}=-x_0x_1^2x_2 (y_0-2y_2 ),&
 \end{aligned}\nonumber\\
 z_{22}= -x_0x_1x_2 \bigl(y_0y_2 (y_0y_2-2x_1 )+4x_0 \bigl(4x_2-y_2^2 \bigr)-4x_2y_0^2 \bigr),\nonumber\\
z_{23}= x_0x_1 (4y_2x_0+y_0 (2x_1-y_0y_2 ) ) \bigl(4x_0y_2^2-y_0^2y_2^2+4x_1y_0y_2-4x_1^2 \bigr),\nonumber\\
 z_{24}= x_2x_1^3x_0^2.\label{fonctionszi}
\end{gather}
 For $c=(c_1,c_2,c_3)\in \Omega$, we consider the regular map
\begin{align}
 \varphi_c\colon \ \mathbb{F}_c \subset\mathcal{H} \longrightarrow \mathbb{P}^{24},\qquad
 (x_0,\dots,y_2 ) \longmapsto (1,z_1,\dots,z_{24} ),\label{plongementP24}
\end{align}
 where the functions $z_i$ are given by \eqref{fonctionszi}; this map
 is embedding of $\mathbb{F}_c$ in the projective space~$\mathbb{P}^{24}(\mathbb{C})$.

Notice that in this section, the strategy is to embed $\mathbb{F}_c
\cup (\mathcal{D} \ \mbox{\{singular
points\}})\longrightarrow \mathbb{P}^{24}$, then show that the
two independent vector fields $\mathcal{V}_1$, $\mathcal{V}_2$ extend
holomorphically to the closure of this embedding in
$\mathbb{P}^{24}$, therefore proving the algebraic complete
integrability of the \smash{$a_4^{(2)}$} Toda lattice, by the
complex Liouville theorem.

\subsubsection[Embedding in the projective space P\^{}{24} and singularities of the divisor at infinity]{Embedding in the projective space $\boldsymbol{\mathbb{P}^{24}}$\\ and singularities of the divisor at infinity}
 In order to show that the $a_4^{(2)}$ Toda lattice is algebraic completely integrable, we show
that, for $c\in \Omega$, the fiber
 \[\mathbb{F}_c:= \mathbb{F}^{-1}(c)=\bigcap_{i=1}^3 \{m\in \mathcal{H}\mid F_i(m)=c_i\} \]
is an affine part of abelian surface, on which the vector fields
$\mathcal{V}_1$ and $\mathcal{V}_2$ restrict to linear vector
fields. To do this, we must check that $\mathbb{F}_c$: satisfies the
conditions of the complex Liouville theorem (Theorem
\eqref{liouville}).
 Using
 \eqref{plongementP24} and let $t \rightarrow 0$, we see that the coefficients of
 $t^{-1}$ of series~$z_i(t;m_0)$ define an application
 \[\varphi_c^{(0)}\colon \ \Gamma_c^{(0)}\longrightarrow \mathbb{P}^{24}\] given by
 \begin{gather}
 \varphi_c^{(0)}\colon \ (a,e) \mapsto (0:-2:0:0:2a:8d-2a^2:0:0:e:-2a(-4d+a^2):-ae:0: \nonumber\\
 \hphantom{\varphi_c^{(0)}\colon \ (a,e) \mapsto (}{} -a^2e:0:de:-ade:-4e(e+3ac):-4ade:de(4d-a^2):12cde: \nonumber\\
 \hphantom{\varphi_c^{(0)}\colon \ (a,e) \mapsto (}{} 4ed(e+3ac):2de^2:4de(3a^2c+ae-12cd): \nonumber\\
 \hphantom{\varphi_c^{(0)}\colon \ (a,e) \mapsto (}{} -16ae(3ac+2e)(3ac+e):e^3d),\label{plongementphi0}
 \end{gather}
which is, for $c\in \Omega$, an embedding of the affine curve
$\Gamma_c$. Similarly, the series~$z_i(t;m_1)$ and~$z_i(t;m_2)$ define two embedding
\smash{$\varphi_c^{(1)}$} and \smash{$\varphi_c^{(2)}$} of the
affine curves \smash{$\Gamma_c^{(1)}$} and
\smash{$\Gamma_c^{(2)}$} respectively in the projective space
$\mathbb{P}^{24}$ given by
 \begin{gather}
\varphi_c^{(1)}\colon \ (a,e) \mapsto \biggl(0:2:1:-a:0:-6c-\frac{a^2}{2}:-6c+\frac{3a^2}{2}:d:-8d+6ac+\frac{a^3}{2}: \nonumber\\
 \hphantom{\varphi_c^{(1)}\colon \ (a,e) \mapsto\biggl(}{}
 0:e:2ad:-8ea:\frac{d\bigl(a^2+12c\bigr)}{2}:0:0:e\bigl(a^2+12c\bigr): \nonumber\\
 \hphantom{\varphi_c^{(1)}\colon \ (a,e) \mapsto\biggl(}{}
 -d\bigl(-16d+12ca+a^3\bigr):-ed:ed:0:-2aed:de\bigl(12c+a^2\bigr): \nonumber\\
 \hphantom{\varphi_c^{(1)}\colon \ (a,e) \mapsto\biggl(}{}
 -\frac{1}{4}e\bigl(12c+a^2\bigr)\bigl(a^4+24ca^2-16e+144c^2\bigr):de^2\biggr),\label{plongementphi1}\\
\varphi_c^{(2)}\colon \ (a,e) \mapsto \bigl(0:0:-2: 0:-2a:0:48d:0:24ad:0:0:4e:0:2ae: \nonumber\\
 \hphantom{\varphi_c^{(2)}\colon \ (a,e) \mapsto\bigl(}{}
 0:-2ce:0:-2a^2e:0:0:-2ce\bigl(-4c+a^2\bigr):0:0: \nonumber\\
 \hphantom{\varphi_c^{(2)}\colon \ (a,e) \mapsto\bigl(}{}
 -8ce\bigl(-4c+a^2\bigr)^2:0\bigr). \label{plongementphi2}
\end{gather}
Remind that in \eqref{plongementphi2} $c=\frac{c_1}{e^2}$ and $d=
\frac{1}{48} \bigl(a^2-c_2-\frac{4c_1}{e^2} \bigr)$, in
\eqref{plongementphi1} $c=\frac{1}{24} \bigl(2a^2-c_2 \bigr)$ and
$e=\frac{c_1}{d^2}$, in \eqref{plongementphi0}
$c=\frac{1}{12} \bigl(4a^2-16d-c_2\big )$ and
$e^2=\frac{c_1}{d^2}$.

Looking at the first three coordinates, we observe that the image
curves by the embedding~$\varphi_c^{ (i )}$ are distinct.
However, they are not complete, so we check if maybe their
closures intersect.

Let us determine the singularities of the divisor at infinity. Let
us denote by $\mathcal{D}_c^{ (0 )}$, $\mathcal{D}_c^{ (1 )}$ and~$\mathcal{D}_c^{(2)}$, respectively, the closures of
\[
\overline{\varphi_c^{ (0 )} \bigl(\Gamma_c^{ (0 )} \bigr)},  \quad \overline{\varphi_c^{ (1 )} \bigl(\Gamma_c^{ (1 )} \bigr)} \quad \text{and} \quad \overline{\varphi_c^{(2)} \bigl(\Gamma_c^{(2)} \bigr)},
 \]
and let \smash{$\mathcal{D}_c = \bigcup_{i=0}^2 \mathcal{D}_c^{ (i )}$}. Let us determine the singularity of
the divisor $\mathcal{D}_c$. To do this, let us substitute the local
parametrization $\varsigma$ around each point at infinity in the
corresponding \smash{$\varphi_c^{ (i )}$} embedding and let
$\varsigma \rightarrow 0$, we find the following leading terms,
where $\epsilon=\epsilon_1=\epsilon_2=\epsilon_3=\pm 1$.

By substituting \eqref{valeurde_inftyvarepsilon_D0},
\eqref{valeurde_inftyvarepsilon1epsilon2_D0} and
\eqref{valeurde_inftyvarepsilon3_D0} in \eqref{plongementphi0} and
considering the first two terms, we have
\begin{gather}
 \varphi_c^{ (0 )} (\infty_\epsilon ) \sim \biggl(0: \dots: 0: 1: 0: 0: 0: 0: 0: 0:- c_3:\frac{\epsilon \sqrt{c_1}}{4} \biggr),\nonumber\\
 \varphi_c^{ (0 )} (\infty_{\epsilon_1\epsilon_2} ) \sim \biggl(0:0:0:0:2:0:\cdots:0: \frac{-c_2+\epsilon_2\delta}{4}:0:\cdots:0: \epsilon_1\epsilon_2\sqrt{c_1}:0:4\epsilon_1\epsilon_2\sqrt{ c_1} \nonumber\\ \hphantom{\varphi_c^{ (0 )} (\infty_{\epsilon_1\epsilon_2} ) \sim \biggl(}{}
 0:0:\frac{\epsilon_2\sqrt{c_1} (\epsilon_1c_2+\epsilon_1\epsilon_2\delta )}{2}:0:0:\epsilon_1\epsilon_2\sqrt{c_1} (\epsilon_1\epsilon_2c_2+\epsilon_1\delta)^2:0 \biggr) , \nonumber\\
 \varphi_c^{ (0 )} (\infty^3_{\epsilon_3} ) \sim \biggl(0:\dots: 0: 1: 0: 0: 0: 0: 0: 0:-c_3:\frac{\epsilon_3\sqrt{c_1}}{4} \biggr).\label{tangentem_0}
\end{gather}
Substituting $\epsilon=\epsilon_1=\epsilon_2=\epsilon_3=\pm 1$ in
\eqref{tangentem_0}, we have
\begin{gather*}
 P_+ := \lim \limits_{p\rightarrow \infty_{+}}\varphi_c^{(0)}(p)= \lim \limits_{p\rightarrow \infty^3_{+}}\varphi_c^{(0)}(p)= \biggl(0:0:0:\dots:0:1:0:\dots:0:-c_3:\frac{ \sqrt{c_1}}{4}\biggr),\\
 P_- := \lim \limits_{p\rightarrow \infty_{-}}\varphi_c^{(0)}(p)= \lim \limits_{p\rightarrow \infty^3_{-}}\varphi_c^{(0)}(p)= \biggl(0:0:0:\dots:0:1:0:\dots:0:-c_3:-\frac{
 \sqrt{c_1}}{4}\biggr),\\
 Q_{++} := \lim \limits_{p\rightarrow
 \infty_{++}}\varphi_c^{(0)}(p)
 = \biggl(0:0:0:0:2:0:0:0:0:\frac{-c_2+\delta}{4}:0:0:0:0:0: \sqrt{c_1}:0:\\ \hphantom{Q_{++}:=\lim \limits_{p\rightarrow\infty_{++}}\varphi_c^{(0)}(p)= \biggl(}{}
 4\sqrt{c_1}:0:0:\frac{\sqrt{c_1}(c_2+\delta)}{2}:0:0:\sqrt{c_1}(c_2+\delta)^2:0\biggr),\\
 Q_{-+} := \lim \limits_{p\rightarrow
 \infty_{-+}}\varphi_c^{(0)}(p)
 = \biggl(0:0:0:0:2:0:0:0:0:\frac{-c_2+\delta}{4}:0:0:0:0:0: -\sqrt{c_1}:\\ \hphantom{ Q_{-+}:= \lim \limits_{p\rightarrow\infty_{-+}}\varphi_c^{(0)}(p)= \biggl(}{}
 0:-4\sqrt{c_1}:0:0:\frac{\sqrt{c_1}(-c_2-\delta)}{2}:0:0:-\sqrt{c_1}(c_2+\delta)^2:0\biggr),\\
 Q_{+-} := \lim \limits_{p\rightarrow
 \infty_{+-}}\varphi_c^{(0)}(p)
 = \biggl(0:0:0:0:2:0:0:0:0:\frac{-c_2-\delta}{4}:0:0:0:0:0: -\sqrt{c_1}:\\ \hphantom{Q_{+-} := \lim \limits_{p\rightarrow\infty_{+-}}\varphi_c^{(0)}(p) = \biggl(}
 0:-4\sqrt{c_1}:0:0:\frac{\sqrt{c_1}(-c_2+\delta)}{2}:0:0:-\sqrt{c_1}(-c_2+\delta)^2:0\biggr),\\
Q_{--} := \lim \limits_{p\rightarrow
 \infty_{--}}\varphi_c^{(0)}(p)
 = \biggl(0:0:0:0:2:0:0:0:0:\frac{-c_2-\delta}{4}:0:0:0:0:0: \sqrt{c_1}:0:\\ \hphantom{Q_{--} := \lim \limits_{p\rightarrow\infty_{--}}\varphi_c^{(0)}(p)= \biggl(}{}
 4\sqrt{c_1}:0:0:\frac{\sqrt{c_1}(c_2-\delta)}{2}:0:0:\sqrt{c_1}(-c_2-\delta)^2:0\biggr).
\end{gather*}
The point $P_+$ is the image of the two points at infinity
$\infty_{+}$ and $\infty^3_{+}$ while $P_-$ is the image of two
points at infinity $\infty_{-}$ and $\infty^3_{-}$. We deduce that
\smash{$\varphi_c^{(0)}$} does not extend into an embedding of~\smash{$\overline{\varphi_c^{(0)}}$}. The curve
\smash{$\mathcal{D}_c^{(0)}$} is therefore singular at the points
$P^{\epsilon}$.

By substituting \eqref{valeurde_inftyvarepsilon_D1},
\eqref{valeurde_infty1_D1} and \eqref{valeurde_infty2_D1} in
\eqref{plongementphi1} and consider the two first terms, we obtain
\begin{gather*}
 \varphi_c^{(1)}\bigl(\infty^1\bigr) \sim (0:0:0:\dots:0:-\varsigma:0:0:0:1:0:\dots:0),\\
\varphi_c^{(1)}\bigl(\infty^2\bigr) \sim \biggl(0:0:0:\dots:0:\frac{1}{16}\varsigma:-\frac{1}{2}\varsigma:0:0:0:0:-\frac{c_2}{64}\varsigma:0:0:0:1:0:\dots:0\biggr),\\
 \varphi_c^{(1)}(\infty_\epsilon) \sim \biggl(0:0:0:\dots:0:-4\varsigma:0:0:0:1:0:0:0:0:0:0:-c_3:\varepsilon\frac{\sqrt{c_1}}{4}\biggr).
\end{gather*}
Letting $\varsigma \rightarrow 0$ and $\epsilon=\pm 1$, we obtain
the points in projective space $\mathbb{P}^{24}$
\begin{gather*}
 P_+ :=\lim \limits_{p\rightarrow \infty_{+}}\varphi_c^{(1)}(p)
 =\biggl(0:0:0:\dots:0:1:0:0:0:0:0:0:-c_3:\frac{\sqrt{c_1}}{4}\biggr),\\
 P_- :=\lim \limits_{p\rightarrow
\infty_{-}}\varphi_c^{(1)}(p)
 =\biggl(0:0:0:\dots:0:1:0:0:0:0:0:0:-c_3:-\frac{\sqrt{c_1}}{4}\biggr),\\
 T :=\lim \limits_{p\rightarrow \infty^1}\varphi_c^{(1)}
 (p) =\lim \limits_{p\rightarrow \infty^2}\varphi_c^{(1)}(p)=
 (0:0:0:0:\dots:1:0:0:0:0:0:0:0).
\end{gather*}
The points $P_+$ and $P_-$ are distinct but $T$ is the image of two
points at infinity $\infty_{1}$ and $\infty_{2}$. We then deduce
that \smash{$\varphi_c^{(1)}$} does not extend into an embedding
of \smash{$\overline{\varphi}_c^{(1)}$}. The curve
\smash{$\mathcal{D}_c^{(1)}$} is therefore singular at the point
$T$.
 By substituting \eqref{valeurde_inftyvarepsilonm2} and
\eqref{valeurde_inftym2} in \eqref{plongementphi2} and considering
that the first three terms and the fact that
$\epsilon_1^2=\epsilon_2^2=1$, we obtain the following points in the
projective space $\mathbb{P}^{24}$:
\begin{gather*}
 \varphi_c^{(2)}(\infty_{++}) \sim
 \biggl(0:0:0:0:2:0:0:0:\frac{-c_2+\delta}{4}:0:0:0:0:0:0: \sqrt{c_1}:0:4\sqrt{c_1}:\\ \hphantom{\varphi_c^{(2)}(\infty_{++}) \sim\biggl(}{}
 0:0:\frac{\sqrt{c_1}(c_2+\delta)}{2}:0:0:\sqrt{c_1}(c_2+\delta)^2:0\biggr),\\
 \varphi_c^{(2)}(\infty_{-+})\sim
 \biggl(0:0:0:0:2:0:0:0:\frac{-c_2+\delta}{4}:0:0:0:0:0:0: -\sqrt{c_1}:0:-4\sqrt{c_1}:\\ \hphantom{\varphi_c^{(2)}(\infty_{-+})\sim\biggl(}{}
 0:0:\frac{\sqrt{c_1}(-c_2-\delta)}{2}:0:0:-\sqrt{c_1}(c_2+\delta)^2:0\biggr),\\
 \varphi_c^{(2)}(\infty_{+-})\sim
 \biggl(0:0:0:0:2:0:0:0:\frac{-c_2-\delta}{4}:0:0:0:0:0:0: -\sqrt{c_1}:0:-4\sqrt{c_1}:\\ \hphantom{\varphi_c^{(2)}(\infty_{+-})\sim\biggl(}{}
 0:0:\frac{\sqrt{c_1}(-c_2+\delta)}{2}:0:0:-\sqrt{c_1}(-c_2+\delta)^2:0\biggr),\\
 \varphi_c^{(2)}(\infty_{--})\sim
 \biggl(0:0:0:0:2:0:0:0:\frac{-c_2-\delta}{4}:0:0:0:0:0:0: \sqrt{c_1}:0:4\sqrt{c_1}:\\ \hphantom{\varphi_c^{(2)}(\infty_{--})\sim\biggl(}{}
 0:0:\frac{\sqrt{c_1}(c_2-\delta)}{2}:0:0:\sqrt{c_1}(-c_2-\delta)^2:0\biggr),\\
 \varphi_c^{(2)}(\infty) \sim (0:0:0:\dots:0:-\varsigma:0:0:0:1:0:0:0:0:0:0:0).
\end{gather*}
Hence, making $\varsigma \rightarrow 0$, we obtain
\begin{gather*}
 Q_{++} :=\lim \limits_{p\rightarrow
 \infty_{++}}\varphi_c^{(2)}(p)
 = \biggl(0:0:0:0:2:0:0:0:\frac{-c_2+\delta}{4}:0:0:0:0:0:0: \sqrt{c_1}:0:\\ \hphantom{ Q_{++} :=\lim \limits_{p\rightarrow\infty_{++}}\varphi_c^{(2)}(p) =\biggl(}{}
 4\sqrt{c_1}:0:0:\frac{\sqrt{c_1}(c_2+\delta)}{2}:0:0:\sqrt{c_1}(c_2+\delta)^2:0\biggr),\\
 Q_{-+} :=\lim \limits_{p\rightarrow
 \infty_{-+}}\varphi_c^{(2)}(p)
 = \biggl(0:0:0:0:2:0:0:0:\frac{-c_2+\delta}{4}:0:0:0:0:0:0: -\sqrt{c_1}:\\ \hphantom{Q_{-+} :=\lim \limits_{p\rightarrow\infty_{-+}}\varphi_c^{(2)}(p) =\biggl(}{}
 0:-4\sqrt{c_1}:0:0:\frac{\sqrt{c_1}(-c_2-\delta)}{2}:0:0:-\sqrt{c_1}(c_2+\delta)^2:0\biggr),\\
 Q_{+-} :=\lim \limits_{p\rightarrow
 \infty_{+-}}\varphi_c^{(2)}(p)
 = \biggl(0:0:0:0:2:0:0:0:\frac{-c_2-\delta}{4}:0:0:0:0:0:0: -\sqrt{c_1}:\\ \hphantom{Q_{+-} :=\lim \limits_{p\rightarrow\infty_{+-}}\varphi_c^{(2)}(p)=\biggl(}{}
 0:-4\sqrt{c_1}:0:0:\frac{\sqrt{c_1}(-c_2+\delta)}{2}:0:0:-\sqrt{c_1}(-c_2+\delta)^2:0\biggr),\\
 Q_{--} :=\lim \limits_{p\rightarrow
 \infty_{--}}\varphi_c^{(2)}(p)
 = \biggl(0:0:0:0:2:0:0:0:\frac{-c_2-\delta}{4}:0:0:0:0:0:0: \sqrt{c_1}:0:\\ \hphantom{Q_{--} :=\lim \limits_{p\rightarrow\infty_{--}}\varphi_c^{(2)}(p)=\biggl(}{}
 4\sqrt{c_1}:0:0:\frac{\sqrt{c_1}(c_2-\delta)}{2}:0:0:\sqrt{c_1}(-c_2-\delta)^2:0\biggr),\\
 T :=\lim \limits_{p\rightarrow \infty}\varphi_c^{(2)}(p)=(0:0:0:\dots:0:1:0:\dots:0).
\end{gather*}
The points $Q_{\epsilon_1\epsilon_2}$, $T$ are all distinct. We
then deduce that~$\varphi_c^{(2)}$ extends into an
embedding of~\smash{$\overline{\varphi}_c^{(2)}$}. The curve
\smash{$\mathcal{D}_c^{(2)}$} is therefore not singular. The
singularities of the curves can be seen from Figure~\ref{completanta4}.
Notice that Figure~\ref{completanta4} shows the divisor and coincides with that
found by M.~Adler and P.~van~Moerbeke (see \cite[Table~2]{al_Moer4}). Hence the conjecture is verify.

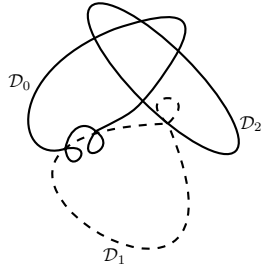
\begin{figure}[t]
 \centering
  \begin{tikzpicture}[smooth,scale=1,thick,font=\small]
\draw[-,dashed] (0.2,0.75) .. controls(-0.3,1.2) and (0.7,1.2) .. (0.2,0.75)
.. controls(1.75,-2.9) and (-3.9,0.62) .. (0.2,0.75);
\draw[-] (-0.5, 2) .. controls(0.7,2.4) and (0.6,2.0) .. (0, 1.2) .. controls(-0.25,0.9) .. (-0.75,0.65) .. controls(-1.1,0.15) and (-0.4,0.4) .. (-0.75,0.65) .. controls(-0.9,0.80) and (-1.1,0.7) .. (-1.15,0.42) .. controls(-1.1,0.0) and (-0.7,0.5) .. (-1.15,0.42)
.. controls(-1.8,0.2) and (-2.0,1.40) .. (-0.5,2);
\draw[rotate=-45] (-0.85,1.05) ellipse (1.35 and 0.4);
\node[scale=0.7] at (1.3,0.75) {$\mathcal{D}_2$};
\node[scale=0.7] at (-0.5,-1.0) {$\mathcal{D}_1$};
\node[scale=0.7] at (-1.75,1.3) {$\mathcal{D}_0$};
\end{tikzpicture}

\vspace{-20mm}

 \caption{Curves completing the invariant surfaces $\mathbb{F}_c$
 of the $a_4^{(2)}$ Toda lattice in abelian surfaces, where~$\mathcal{D}_i$ is the curve
 \smash{$\mathcal{D}_c^{(i)}$}.}\label{completanta4}
 \end{figure}

\subsubsection{The quadratic differential equations and holomorphic}
Now we want to show that the vector field
$(\varphi_c)_{\ast}\mathcal{V}_1$ extends into a
holomorphic vectors field on~$\mathbb{P}^{24}$.
 This is done by
exhibiting the quadratic differential equations in two of the
charts. We will use the following
\begin{Lemma}[{\cite{al_Moer3}}] Let $\mathcal{X}$ be a vector field on
$\mathbb{P}^{N}$ holomorphic in two different maps $(Z_i\neq
0)$ and $(Z_j\neq 0)$. Then $\mathcal{X}$ is
holomorphic over $\mathbb{P}^{N}$. That is to say on any map
$(Z_j\neq 0)$.
\end{Lemma}
\begin{Proposition} The vector field
$(\varphi_c)_{\ast}\mathcal{V}_1$ extends into a
holomorphic vectors field on $\mathbb{P}^{24}$.
\end{Proposition}
\begin{proof}
It suffices to show that the vector field
$(\varphi_c)_{\ast}\mathcal{V}_1$ is holomorphic on two
chart of $\mathbb{P}^{24}$. To do this, let us establish that this
vector field can be written as a quadratic vector field in the chart
$Z_0\neq0$ and $Z_1\neq0$. In the chart $Z_0\neq0$, we obtain the
following result:
\begin{gather*}
 \dot{z}_1= -\frac{1}{2}(c_2-4(z_3+z_4)+z_1(4z_2-z_1)) ,\\
 \dot{z}_2= \frac{1}{4}(3(z_4-z_1z_2)+2(z_2^2-z_3)) ,\\
 \dot{z}_3= \frac{1}{2}z_6-z_2z_3 ,\\
 \dot{z}_4= \frac{1}{2}(z_4(z_1+z_2)-c_2z_2+2z_6-z_5) ,\\
 \dot{z}_5= \frac{1}{2}(z_1z_5-c_3)+2(6z_7-z_9-z_2z_5) ,\\
 \dot{z}_6= -\frac{1}{2}(z_2z_5+z_6(z_1-z_2)+3z_9)+2(4z_7+z_8)-z_3^2 ,\\
 \dot{z}_7= -\frac{1}{2}(z_{11}+2z_2z_7) ,\\
 \dot{z}_8= \frac{1}{2}z_1z_8-z_{10} ,\\
 \dot{z}_9= \frac{1}{2}(z_2(c_2z_3z+z_9-24z_7)+z_5(z_3-z_4))-4z_{11}-z_3z_6 ,\\
 \dot{z}_{10}= \frac{1}{2}(-2z_2z_{10}+z_4z_8+4z_{14}) ,\\
 \dot{z}_{11}= \frac{1}{2}(z_{13}+2(z_3z_7+2z_{14})+z_{11}(z_2-z_1)) ,\\
 \dot{z}_{12}= \frac{1}{2}(2z_3z_{10}-z_5z_8-4z_1z_{14}) ,\\
 \dot{z}_{13}= \frac{1}{2}(z_{13}(z_2-z_1)+z_5z_7-4z_{15}) ,\\
 \dot{z}_{14}= \frac{1}{2}z_1z_{14} ,\\
 \dot{z}_{15}= \frac{1}{2}(z_{18}-2z_7z_8+z_{15}(z_1+z_2)),\\
 \dot{z}_{16}= \ast,\\
 \dot{z}_{17}= \Delta,\\
 \dot{z}_{18}= -\frac{1}{2}z_5z_{14}+z_7z_{10},\\
 \dot{z}_{19}= 4z_7z_{10}-\frac{1}{2}z_{14}(z_1-4z_5) ,\\
 \dot{z}_{20}= \frac{1}{2}(z_{22}+z_2z_{20}-z_4z_{19}+4z_8(z_{13}-2z_{14})),\\
 \dot{z}_{21}= 2z_7(z_{12}+2z_{14})-z_8(z_{13}+2z_{14}),\\
 \dot{z}_{22}= \frac{1}{2}(z_1z_{22}-z_7z_{16})-2(z_{10}z_{13}+6z_1c_1) ,\\
 \dot{z}_{23}= \star ,\\
 \dot{z}_{24}= \frac{1}{2}(z_{24}(z_1-4z_2))
\end{gather*}
with
\begin{gather*}
 \ast=\frac{1}{4}(z_{16}(z_1-4z_2)+z_8(32z_7+2z_8-c_2z_4)+4(2z_{19}+z_5z_{10})), \\
 \star=\frac{1}{2}(z_{23}(z_1+z_2)+5z_{16}(8z_{11}+c_3z_2))+128z_{14}(z_{13}+2z_{14})+c_3z_8(c_3-32z_7)\\ \hphantom{\star=}{}
 +64c_1z_1(2z_2-z_1)-4c_3z_{19},\\
 \Delta =\frac{1}{2}(4(3z_{18}+2z_1z_{15})-z_1z_{17}-z_4z_{13}-z_7(8z_8+c_3)).
\end{gather*}
 For the second chart $Z_1\neq0$, we put $s_i=z_i/z_1$, for
all $i=0,\dots,24$. Then the quadratic differential equations take
the following form:
\begin{gather*}
 \dot{s}_0= \frac{1}{2}(s_0(c_2s_0-4(s_3+s_4))+4s_2-1) , \\
 \dot{s}_1= 0 , \\
 \dot{s}_2= \frac{1}{2}(s_0(c_2s_2-s_5-4s_6)+s_2(4s_3+s_4)-s_4) , \\
 \dot{s}_3= \frac{1}{2}(s_0(c_3s_0-32s_7+4s_9)-s_5+s_6+2s_2s_5) , \\
 \dot{s}_4= \frac{1}{2}\bigl(s_0(8s_8+6s_9)+s_5(6s_2-1)+s_4^2\bigr), \\
 \dot{s}_5= 4s_0(s_{10}+s_{11})-s_9+8s_2s_7-s_4s_5 , \\
 \dot{s}_6= \frac{1}{2}s_0(16s_{10}-24s_{11}-c_2s_6)+s_2(2c_2s_3-16s_7-s_9)-s_3(s_5+s_6)-2s_8+s_9 , \\
 \dot{s}_7= s_0(s_{13}-2s_{14})-\frac{1}{2}s_{11} , \\
 \dot{s}_8= \frac{1}{2}(s_0(c_2s_8-32s_{14}))+s_{10}(4s_2-1)-6s_3s_8, \\
 \dot{s}_9= \frac{3}{4}(s_0(c_3s_4+32s_{14}+16s_{13})-c_3s_2) \\ \hphantom{\dot{s_9}=}{}
 +\frac{1}{2}\bigl(s_2(2(c_3s_2+8s_{11}-2s_{10})-c_2s_6) +s_4(s_9-4s_8)+s_5^2+2s_6^2-16s_{11}\bigr), \\
 \dot{s}_{10}= -\frac{1}{2}s_0s_{16}-2s_{14}-s_5s_8 , \\
 \dot{s}_{11}= -2s_{14}+4s_0s_{15}+s_{13}(s_2-1)-\frac{1}{2}s_7(c_2+6s_5) , \\
 \dot{s}_{12}=\frac{1}{8}(c_2s_4s_8+s_{16}(4s_2+1))-s_8(4s_7+3s_8), \\
 \dot{s}_{13}= \frac{1}{2}(s_4s_{13}+s_{17}+s_7(3c_3s_0-48s_7-4s_8))-6s_0s_{17} , \\
 \dot{s}_{14}= -\frac{1}{2}(s_0s_{19}-4s_7s_8), \\
 \dot{s}_{15}= \frac{1}{2}(s_0s_{20}+s_4s_{15}-s_5s_{14}+4s_7s_{10}), \\
 \dot{s}_{16}= 2(8s_{18}+3s_{19})+\frac{1}{2}s_{16}(c_2s_0-4s_{4})+2s_{10}(8s_8+c_3s_0+32s_7) \\ \hphantom{\dot{s_{16}}=}{}
 -s_5(c_2s_8+4s_{12})-32s_0s_{21}, \\
 \dot{s}_{17}= 2s_{14}(s_{5}-2s_{6})+s_{11}(4s_{8}-c_2s_{3})+\frac{1}{2}s_{7}(c_2(2s_5-c_2s_2)+c_3+8s_{10})-3s_5s_{13} \\ \hphantom{\dot{s_{17}}=}{}
-2s_0s_{20}+\frac{1}{2}s_4s_{17} , \\
 \dot{s}_{18}= \frac{1}{2}(s_3s_{19}+s_{21})+8c_1s_0^2+2s_7s_{12}+s_8(2s_{14}-s_{13}), \\
 \dot{s}_{19}= \frac{1}{2}c_2s_0s_{19}-2s_0s_{22}+24s_8s_{14}+2s_{21} , \\
 \dot{s}_{20}= \frac{1}{2}(s_{22}+s_4s_{20}-c_2s_2s_{19})+s_6s_{19}-s_7s_{16}-8s_8s_{15}, \\
 \dot{s}_{21}= 4c_1s_0+2s_{10}s_{13}-s_7s_{16}, \\
 \dot{s}_{22}= \frac{1}{2}(s_0(c_3s_{19}-16s_{24}))+8\bigl(s_{10}s_{15}-8s_{14}^2\bigr)+4c_1, \\
 \dot{s}_{23}= \frac{1}{2}(s_{23}(s_4+c_2s_0)-c_3(s_{8}(3c_3+40s_{10})-2s_3s_{16}-4s_{19}(4s_2+3)))-96s_{24}(1+2s_2) \\ \hphantom{\dot{s_{23}}=}{}
 +8s_{16}(4s_{13}+10s_{14}-c_2s_7)-192c_1s_4+160s_{11}s_{19}-96s_{24}, \\
 \dot{s}_{24}= \frac{1}{2}(s_{10}s_{20}-s_{12}s_{19}-2s_4s_{24}).
\end{gather*}
This shows that the vector field
$(\varphi_c)_{\ast}\mathcal{V}_1$ extends into a linear
vector field $\overline{\mathcal{V}}_1$ on $\mathbb{P}^{24}$.
\end{proof}

Thus, the item four of the complex Liouville theorem (Theorem
\eqref{liouville}) is satisfied. We now show that the integral
curves of $\overline{\mathcal{V}}_1$ that start at the three
singular points go into the affine immediately
\begin{Proposition}
 The flow $\Phi_t$ of the vector field $\overline{\mathcal{V}}_1$ on $\mathbb{P}^{24}$ coming from the points of $\mathcal{D}_c$ is sent into the affine part $\varphi_c(\mathbb{F}_c)$.
\end{Proposition}
\begin{proof}
Note that this is automatic for the points of $\varphi_c\bigl(\mathbf{\Gamma}_c^{(0)}\bigr)$, $\varphi_c\bigl(\mathbf{\Gamma}_c^{ (1)}\bigr)$ and $\varphi_c\bigl(\mathbf{\Gamma}_c^{(2)}\bigr)$. We then just need to check for the singular points $Q_{\varepsilon_1\varepsilon_2}$, $P_{\varepsilon}$ and $T$ in $\mathbb{P}^{24}$.

 Consider the $4$ points $Q_{\varepsilon_1\varepsilon_2}$
intersections of $\mathcal{D}_c^{(0)}$ and
$\mathcal{D}_c^{(2)}$. From
\eqref{valeurde_inftyvarepsilon1epsilon2_D0} and
\eqref{plongementphi0} follow that the leading coefficient of
\smash{$z_4\bigl(t;\mathbf{\Gamma}_c^{(0)}\bigr)$} has a pole
for $\varsigma=0$, that is maximal with the leading coefficient of
\smash{$z_9\bigl(t;\mathbf{\Gamma}_c^{(0)}\bigr)$} and thus the
function $z_4$ defines a map at these points.
 Let us show~then that \smash{$\lim_{p\rightarrow
 \infty_{\epsilon_1\epsilon_2}}\frac{1}{z_4}\varphi_c^{(0)}\neq0$}.

 Using \eqref{balancemainm0}, we get
 \[
 z_4(t;m_0)=\frac{2a}{t}-4d+(-e+2ad-2ac)t+\frac{1}{6}\bigl(ae-4\bigl(2d^2+da^2\bigr)+4cd\bigr)t^2+O\bigl(t^3\bigr).
 \]
 The first terms of the inverse of this series are then
 given by
\begin{equation}\label{inverse4}
 \frac{1}{z_4(t;m_0)}=\frac{t}{2a}+\frac{d}{a^2}t^2+O\bigl(t^3\bigr).
\end{equation}

Substituting
$e=-\frac{4a^4-32a^2d-a^2c_2+64d^2+4dc_2+c_3}{8a}$, $c=\frac{a^2}{3}-\frac{4d
}{3}-\frac{c_2}{12}$ in the second term and rewriting the
coefficients according to the local parameter $\varsigma$ around
$\infty_{\epsilon_1\epsilon_2}$ using
\eqref{valeurde_inftyvarepsilon1epsilon2_D0}, we obtain
\[\lim \limits_{\varsigma\rightarrow
 0}\frac{1}{z_4}(t,\varsigma)=\frac{1}{4}t^2+O\bigl(t^3\bigr)\neq0,\]
 which prove that the integrals curves of
 $\overline{\mathcal{V}}_1$ which start from the points $Q_{\epsilon_1\epsilon_2}$
 are immediately sent in the affine part $\varphi_c(\mathbb{F}_c)$.

 For the point $P_{\epsilon}$, intersection points of $\mathcal{D}_c^{(0)}$ and $\mathcal{D}_c^{(1)}$,
 the residue having the largest pole among the residues of
 \smash{\raisebox{-0.3pt}{$z_0\bigl(t;\mathbf{\Gamma}_c^{(1)}\bigr),\dots,z_{24}\bigl(t;\mathbf{\Gamma}_c^{(1)}\bigr)$}}
 is \smash{$z_9\bigl(t;\mathbf{\Gamma}_c^{(1)}\bigr)$}, the function $z_{15}$ defines a local chart around this point. Consider the balances $x(t;m_1)$
 \ref{balancemainem1} and rewriting $a$, $b$, $c$, $d$ of
 \smash{$1/z_9\bigl(t;\mathbf{\Gamma}_c^{(1)}\bigr)$} depending on the local parameter $\varsigma$ around $\infty_{\epsilon}$ using \smash{$c = \frac{1}{12}a^2- \frac{1}{24}c_2$, $e =\frac{c_1}{d^2}$} and \eqref{valeurde_inftyvarepsilon_D1},
 we obtain \[\lim \limits_{\varsigma\rightarrow
 0}\frac{1}{z_9}(t,\varsigma)=\frac{1}{12}t^4+O\bigl(t^6\bigr)\neq0.\]
 Consider the balances $x(t;m_0)$
 \eqref{balancemainm0} and the local parametrization \eqref{valeurde_inftyvarepsilon_D0}, we also obtain
 \[
 \lim \limits_{\varsigma\rightarrow 0}\frac{1}{z_8}(t,\varsigma)=\frac{1 }{12}t^4+O\bigl(t^6\bigr)\neq0.
 \]
 Thus, the different limits found are different from zero, which also shows that the integral curve of $\overline{\mathcal{V}}_1$ which starts from the points $P_{\epsilon}$ are
 immediately sent in the affine part $\varphi_c(\mathbb{F}_c)$.

 For the point $T$, intersection point of $\mathcal{D}_c^{(2)}$ and $\mathcal{D}_c^{( 1)}$,
 the only non-zero coordinate corresponds to the function \smash{$z_{17}\bigl(t;\mathbf{\Gamma}_c^{(1)}\bigr)$}, the function $z_{17}$
 defines a local map around this point. Consider the balances $x(t;m_1)$ \eqref{balancemainem1}
 and $x(t;m_2)$ \eqref{balancemainem2} by writing $a$, $c$, $d$, $e$ of \smash{$1/z_{17}\bigl(t;\mathbf {\Gamma}_c^{(1)}\bigr)$} depending on the local parameter $\varsigma$ around $\infty_{\epsilon}$ and
 using \smash{$c = \frac{ 1}{12}a^2-\frac{1}{24}c_2$, $e =\frac{c_1}{d^2}$} and \eqref{valeurde_inftyvarepsilon_D1}, we get
 \[
 \lim \limits_{\varsigma\rightarrow 0}\frac{1}{z_{17}}(t,\varsigma)=\frac{1}{2304}t^7+O\bigl(t^8\bigr)\neq0.
 \]
 Considering the balances $x(t;m_2)$
 \eqref{balancemainem2} and local parametrization
 \eqref{valeurde_inftym2}, we also get
 \[
 \lim\limits_{\varsigma\rightarrow 0}\frac{1}{z_{17}}(t,\varsigma)=-\frac{1 }{1152}t^7+O\bigl(t^8\bigr)=-2\times\biggl(\frac{1}{2304}t^7\biggr)+O\bigl(t^8 \bigr)\neq0.
 \]
 Thus, the different limits found are different from zero,
 which also shows that the integral curve of $\overline{\mathcal{V}}_1$ which starts from the points $T$ are
 immediately sent in the affine part $\varphi_c(\mathbb{F}_c)$.

 Thus, the flow of the vector field $\overline{\mathcal{V}}_1$
 starting from each of the points of \smash{$\mathcal{D}_c^{(0)}\cup
 \mathcal{D}_c^{(1)}\cup \mathcal{D}_c^{(2)}$} goes into the affine part of $\varphi_c(\mathbb{ F}_c)$.

 In order to complete the proof of algebraic integrability, it is necessary to show that there exist no other divisors passing through the points $Q_{\epsilon_1\epsilon_2}$, $P_{\epsilon}$ and $T$.

 For the four intersections points $Q_{\epsilon_1\epsilon_2}$ of \smash{$\mathcal{D}_c^{(0)}$} and
 \smash{$\mathcal{D}_c^{(2)}$}, by rewriting the coefficients of the right-hand side of~(\ref{inverse4}) with respect to the local parameter $\varsigma$ in the neighborhood of the points \smash{$\infty_{\epsilon_1\epsilon_2}\in\overline{\Gamma}^{(0)}_c$}, we find
 \[
 \frac{1}{z_4\bigl(t;\Gamma^{(0)}_c\bigr)}=\frac{1}{4}\bigl(2\varsigma t+t^ 2\bigr)+O\bigl(t^3,\varsigma t^2\bigr),
 \]
 which shows that the multiplicity of \smash{$\frac{1}{z_4}$} at each of these points is equal to $2$, which coincides~with the sum of the orders of zero of \smash{$\frac {1}{z_4}$} on each of the divisors so there are no other divisors passing through the points $Q_{\epsilon_1\epsilon_2}$.

 For the points $P_{\epsilon}$ and $T$ of the divisor $\mathcal{D}_c$, obtained from the points at infinity $\infty_{\epsilon}$, $\infty_{1}$ and~$\infty_{2}$, let us check that the degree of $\mathcal{D}_c$ which is $3$ is indeed equal to the degree of \smash{$\overline{\varphi_c(\mathbb{F}_c)}\setminus \varphi_c(\mathbb{F}_c)$} at these points. As the vector field $\overline{\mathcal{V}}_1$ is only tangent to one of the branches (transverse to the other) of~$\mathcal{D}_c^{1}$ passing through these two points, we do the expansion along the non-tangent branch. As~the~function $z_{17}$ defines a map at point $T$, we just need to substitute \eqref{valeurde_inftyvarepsilon_D1}
 into its inverse series,
 \[
 \frac{1}{z_{17}\bigl(t;\mathbf{\Gamma}_c^{(1)}\bigr)}=-\frac{t}{d\bigl(-16d+12ac+a^3\bigr)}+ O\bigl(t^2\bigr)=-\frac{2t}{d\bigl(-32d-ac_2+4a^3\bigr)}+ O\bigl(t^2\bigr),
 \]
 which leads to
 \[
 \frac{1}{z_{17}\bigl(t;\mathbf{\Gamma}_c^{(1)}\bigr)}=-\frac{1}{ 4\sqrt{c_1}}\varsigma t+ O\bigl(t^3\bigr),
 \]
 thus showing that there are no other divisors passing through $T$.

For points $P_{\epsilon}$, we do this by calculating
 the first terms of the series \smash{$1/z_{8}\bigl(t;\mathbf{\Gamma}_c^{(0)}\bigr)$} using Laurent's solution $x(t;m_0)$ then we express the free parameters according to the local parameter~$\varsigma$ in a~neighborhood of \eqref{valeurde_inftyvarepsilon_D0}. The resulting series in $\varsigma$ and $t$ should
 start with monomials of degree~$3$ because the point $P_{\epsilon}$ has multiplicity $2$ and $1$ on the divisors \smash{$\mathcal{D}_c^ {(0)}$} and \smash{$\mathcal{D}_c^{(1)}$}, respectively, and the function $z_{8}$ has a simple pole on each of these divisors. We have
 \[
 \frac{1}{z_{8}\bigl(t;\mathbf{\Gamma}_c^{(0)}\bigr)}=\frac{1}{e}t-\frac{at^2}{e} + O\bigl(t^3\bigr),
 \]
 which leads to
 \[
 \frac{1}{z_{8}\bigl(t;\mathbf{\Gamma}_c^{(0)}\bigr)}=2\varsigma^2 t^ 2+O\bigl(\varsigma^3,\varsigma t^3\bigr),
 \] thus showing that there are no other passing divisors
 by $P_{\epsilon}$. Furthermore this also shows that $P_{\epsilon}$
 is an ordinary double point for the divisor \smash{$\mathcal{D}_c^{(0)}$}. So, there are no other divisors in \smash{$\overline{\varphi_c(\mathbb{F}_c)}\setminus
 \varphi_c(\mathbb{F}_c)$} besides the divisors \smash{$\mathcal{D}_c^{(0)}$, $\mathcal{D}_c^{(1)}$} and
 \smash{$\mathcal{D}_c^{(2)}$} already found.
\end{proof}

 The Liouville complex theorem conditions being
 satisfied, it follows that for $c\in \Omega$, the
 projective variety \smash{$\overline{\varphi_c(\mathbb{F}_c)}=
 \varphi_c(\mathbb{F}_c)\cup \mathcal{D}_c$} is an abelian surface and the restrictions of vector fields~$\overline{\mathcal{V}}_1$
 and~$\overline{\mathcal{V}}_2$ to these abelian surfaces are linear.
 Since \smash{$\overline{\varphi_c(\mathbb{F}_c)}$} contains a smooth curve of genus~$2$, it is the Jacobian of this curve.
 We have therefore proved the following theorem.

\begin{Theorem}
Let $(\mathcal{H}, \{\cdot,\cdot\},\mathbf{F})$ be an
integrable system  describing the $a_4^{(2)}$ Toda lattice, where
$\mathbf{F}=(F_1,F_2,F_3)$ and $\{\cdot,\cdot\}$ are
given, respectively, by \eqref{invar1a4} and \eqref{matrij2} with
commuting vector fields~\eqref{systa4}.
\begin{itemize}\itemsep=0pt
 \item [$(i)$] $(\mathcal{H}, \{\cdot,\cdot\},\mathbf{F})$ is a weight homogeneous
 algebraical completely integrable system.
 \item [$(ii)$] For $c\in \Omega$, the fiber $\mathbb{F}_c$ of its
 momentum map is completed in an abelian surface
 $\mathbb{T}^2_c$ $\bigl($the Jacobian of the hyperelliptic curve $($of genus two$)$ \smash{$\overline{\Gamma}_c^{(2)} \bigr)$}
 by the addition of a singular~divisor~$\mathcal{D}_c$
 composed of three irreducible components: \smash{$\mathcal{D}_c^{(0)}$} defined by
 \begin{gather*}
 \Gamma_c^{(0)}\colon \ 16d^2a^8-\bigl(256d^3+8d^2c_2\bigr)a^6+ \bigl(1536d^2+96dc_2+8c_3+c_2^2\bigr)d^2a^4\\
 \hphantom{\Gamma_c^{(0)}\colon}{} \
 -\bigl(\bigl(8\bigl(8c_3+48dc_2+c_2^2+512d^2\bigr)d+2c_2c_3\bigr)d^2+64c_1\bigr)a^2 \\
 \hphantom{\Gamma_c^{(0)}\colon}{} \
 +\bigl(8d\bigl(c_2c_3+ 16dc_3+64d^2c_2 +512d^3+2dc_2^2\bigr)+c_3^2\bigr)d^2=0,
\end{gather*}
 and
 \smash{$\mathcal{D}_c^{(1)} $} defined by
\begin{equation*}
 \Gamma_c^{(1)}\colon \ 256ad^3-\bigl(\bigl(4a^2-c_2\bigr)^2-16c_3\bigr)d^2+64c_1=0,
\end{equation*}
two singular curves of respective genus $3$ and $4$ and one
 smooth curve and \smash{$\mathcal{D}_c^{(2)}$} defined by
\begin{equation*}
 \Gamma_c^{(2)}\colon \ e^4a^4-\bigl(8c_1+c_2e^2\bigr)a^2e^2-64e^5+4e^2c_1c_2+4c_3e^4+16c_1^2=0
\end{equation*}
 of genus $2$ and isomorphic to \smash{$\overline{\Gamma}_c^{(2)} $}. The curves intercept each other as indicated in Figure~$\ref{completanta4}$.
\end{itemize}
\end{Theorem}
\section[Geometry of the a\_4\^{}(2) Toda lattice: holomorphic differentials forms]{Geometry of the $\boldsymbol{a_4^{(2)}}$ Toda lattice:\\ holomorphic differentials forms}\label{sec4}
In this section, we show following an idea by Luc Haine (see
\cite{Haine}) how the holomorphic differentials are used to
determine the tangency locus of the vector field $\mathcal{V}_1$ on
the divisor $\mathcal{D}_c$, we compute the holomorphic
differentials $\omega_1$ and $\omega_2$ on the three irreducible
components of $\mathcal{D}_c$ that come from the differentials
${\rm d}t_1$ and ${\rm d}t_2$ on the abelian surface $\mathbb{T}^2_c$. We know
that all irreducible components have multiplicity $1$.

 Let us calculate the holomorphic differential forms $\omega_1$ and
 $\omega_2$ on the divisor $\mathcal{D}_c$ which come from the differentials ${\rm d}t_1$ and ${\rm d}t_2$ on the abelian surface
 $\mathbf{T}^2_c$. Let $y_0:=z_1$ and $y:=z_4$ be
 restricted to \smash{$\mathcal{D}_c^{(0)}$}, the first coefficients of their Laurent series are given by
 \begin{gather*}
y_0^{(0)}= -2, \qquad  y_0^{(1)}= 2c,\qquad
 y^{(0)}= 2a, \qquad  y^{(1)}=-4d,
 \end{gather*}
and using \eqref{balancemainm0} and \eqref{syst3},
 we have \[\mathcal{V}_2\biggl[\frac{1}{z_1}\biggr]_{|\mathcal{D}_c^{(0)}}=-2a^2+8d , \qquad
 \mathcal{V}_2\biggl[\frac{z_4}{z_1}\biggr]_{|\mathcal{D}_c^{(0)}}=-4ea-24cd.\]
It then follows that
\begin{align*}
 \delta &{}= \frac{1}{\bigl(y_0^{(0)}\bigr)^2}
 \begin{vmatrix}
 y_0^{(0)} & \mathcal{V}_2\bigl[\frac{1}{y_0}\bigr]_{|\mathcal{D}_c^{(0)}} \\[1mm]
 y_0^{(0)}y^{(1)}-y^{(0)}y_0^{(1)} & \mathcal{V}_2\bigl[\frac{z_4}{y_0}\bigr]_{|\mathcal{D}_c^{(0)}}
 \end{vmatrix}
 = \frac{1}{4}
 \begin{vmatrix}
 -2 & -2a^2+8d \\
 8d-4ac & -4ea-24cd \\
 \end{vmatrix}\\
 &{}= 2ea+12cd+4a^2d-2a^3c-16d^2+8cad.
\end{align*}
The holomorphic differential forms ${\rm d}t_1$ and ${\rm d}t_2$ restricted to
\smash{$\mathcal{D}_c^{(0)}$} are given by
\begin{gather*}
 \omega_1 = \frac{1}{\delta y_0^{(0)}}{\rm d}\Biggl(\frac{y^{(0)}}{y_0^{(0)}}\Biggr)=-\frac{{\rm d}a}{\delta }
 \qquad \text{and} \qquad
 \omega_2 = -\frac{1}{\delta }\mathcal{V}_2\biggl[\frac{1}{y_0}\biggr]_{|\mathcal{D}_c^{(0)}}
 {\rm d}\Biggl(\frac{y^{(0)}}{y_0^{(0)}}\Biggr)=-\frac{2a^2-8d}{\delta}{\rm d}a.
\end{gather*}

 To determined differential forms on the divisor
\smash{$\mathcal{D}_c^{(1)}$}, consider the functions $y_0:=z_1$
and $y:=z_3$, restricted to \smash{$\mathcal{D}_c^{( 1)}$}, the
first coefficients of their Laurent series are given by
 \begin{gather*}
  y_0^{(0)}= 2, \qquad   y_0^{(1)}= a,\qquad
  y^{(0)}= -a, \qquad   y^{(1)}=-3c+\frac{a^2}{4},
 \end{gather*}
 and using \eqref{balancemainem1} and \eqref{syst3},
 we have \[\mathcal{V}_2\biggl[\frac{1}{z_1}\biggr]_{|\mathcal{D}_c^{(1)}}=-6c-\frac{a^2}{2}, \qquad
 \mathcal{V}_2\biggl[\frac{z_3}{z_1}\biggr]_{|\mathcal{D}_c^{(1)}}=-2e+4ad+18c^2-3a^2c-\frac{3a^4}{8},\]
then
\begin{align*}
 \delta &{}= \frac{1}{\bigl(y_0^{(1)}\bigr)^2}
 \begin{vmatrix}
 y_0^{(0)} & \mathcal{V}_2\bigl[\frac{1}{y_0}\bigr]_{|\mathcal{D}_c^{(1)}} \\[1mm]
 y_0^{(0)}y^{(1)}-y^{(0)}y_0^{(1)} & \mathcal{V}_2\bigl[\frac{z_3}{y_0}\bigr]_{|\mathcal{D}_c^{(1)}}
 \end{vmatrix}
 \\
 &{}= \frac{1}{4}
 \begin{vmatrix}
 2 & -6c-\frac{a^2}{2} \\
 -6c+\frac{3a^2}{2} & -2e+4ad+18c^2-3a^2c-\frac{3a^4}{8} \\
 \end{vmatrix}
 = 2ad-e.
\end{align*}
The holomorphic differential forms ${\rm d}t_1$ and ${\rm d}t_2$ restricted to
\smash{$\mathcal{D}_c^{(1)}$} are given by
\begin{gather*}
 \omega_1 = \frac{1}{\delta y_0^{(0)}}{\rm d}\Biggl(\frac{y^{(0)}}{y_0^{(0)}}\Biggr)=-\frac{{\rm d}a}{4\delta }
 \qquad \text{and} \qquad
 \omega_2 = -\frac{1}{\delta }\mathcal{V}_2\biggl[\frac{1}{y_0}\biggr]_{|\mathcal{D}_c^{(1)}}
 {\rm d}\Biggl(\frac{y^{(0)}}{y_0^{(0)}}\Biggr)=-\frac{a^2+12c}{4\delta}{\rm d}a.
\end{gather*}

 For calculating differential forms on the divisor
\smash{$\mathcal{D}_c^{(2)}$}, consider the functions $y_0:=z_2$
and $y:=z_4$ restricted to \smash{$\mathcal{D}_c^{(2 )}$}, the
first coefficients of their Laurent series are given by
 \begin{gather*}
  y_0^{(0)}= -2, \qquad   y_0^{(1)}= 2d,\qquad
  y^{(0)}= -2a, \qquad   y^{(1)}=-4c,
 \end{gather*}
 and using \eqref{balancemainem2} and \eqref{syst3},
 we have \[\mathcal{V}_2\biggl[\frac{1}{z_2}\biggr]_{|\mathcal{D}_c^{(2)}}=2c-\frac{a^2}{2}, \qquad
 \mathcal{V}_2\biggl[\frac{z_4}{z_2}\biggr]_{|\mathcal{D}_c^{(2)}}=16e-96cd.\]
It then follows that
\begin{align*}
 \delta &{}= \frac{1}{\bigl(y_0^{(1)}\bigr)^2}
 \begin{vmatrix}
 y_0^{(0)} & \mathcal{V}_2\bigl[\frac{1}{y_0}\bigr]_{|\mathcal{D}_c^{(2)}} \\[1mm]
 y_0^{(0)}y^{(1)}-y^{(0)}y_0^{(1)} & \mathcal{V}_2\bigl[\frac{z_3}{y_0}\bigr]_{|\mathcal{D}_c^{(2)}}
 \end{vmatrix}
 = \frac{1}{4}
 \begin{vmatrix}
 -2 & 2c-\frac{a^2}{2} \\
 8c+4ad & 16e-96cd \\
 \end{vmatrix}\\
 &{}= 48cd-8e+ca^2+\frac{a^3d}{2}-4c^2-2acd.
\end{align*}
The holomorphic differential forms ${\rm d}t_1$ and ${\rm d}t_2$ restricted to
\smash{$\mathcal{D}_c^{(2)}$} are given by
\begin{gather*}
 \omega_1 = \frac{1}{\delta y_0^{(0)}}{\rm d}\Biggl(\frac{y^{(0)}}{y_0^{(0)}}\Biggr)=-\frac{{\rm d}a}{2\delta }
 \qquad \text{and} \qquad
 \omega_2 = -\frac{1}{\delta }\mathcal{V}_2\biggl[\frac{1}{y_0}\biggr]_{|\mathcal{D}_c^{(2)}}
 {\rm d}\Biggl(\frac{y^{(0)}}{y_0^{(0)}}\Biggr)=-\frac{4c-a^2}{2\delta}{\rm d}a.
\end{gather*}
The zero of differentials forms $\omega_1$ and $\omega_2$ gives the
tangency points of fields vector $\mathcal{V}_1$ and $\mathcal{V}_2$
respectively.

\begin{Remark}
In this paper, we use the \textsc{Maple} 13 application to develop and
implement algorithms for determining Laurent series, various curves.
We also use this application to determine the Painlev\'{e} divisors,
the $z_i$ functions, embedding to the projective space and give
the different quadratic differential equations of the two chart
$Z_0$ and $Z_1$.
\end{Remark}

\subsection*{Acknowledgements}
We would like to extend our sincere gratitude to Professor Pol
Vanhaecke at University of Poitiers for his particular contributions
in providing clarifications and guidance on our research theme, for
the enriching exchanges and thoughtful advice he generously offered
us throughout this project. We cannot end our acknowledgements
without thanking all the referees of this paper. We wish to express
our thanks to the referees for their valuable helpful comments and
suggestions.

\pdfbookmark[1]{References}{ref}
\LastPageEnding

\end{document}